\documentclass[12pt,a4paper]{article}
\usepackage[utf8]{inputenc}
\usepackage{graphicx}
\usepackage{amsmath}
\usepackage{bm}
\usepackage{amsthm}
\usepackage{amssymb}
\usepackage{textcomp}
\usepackage{siunitx}
\usepackage{subfigure}
\usepackage{titling}
\usepackage{xcolor}
\usepackage{afterpage}
\usepackage{algorithm} 
\usepackage{algpseudocode} 
\usepackage{multirow}
\usepackage{hyperref}
\hypersetup{
    colorlinks=true,
    linkcolor=blue,
    citecolor = blue,
    filecolor=blue,      
    urlcolor=blue,
    pdftitle={Overleaf Example},
    pdfpagemode=FullScreen,
    }
\usepackage{dsfont}
\usepackage{mathtools}
\usepackage[font=small,labelfont=bf]{caption} 
\usepackage{booktabs}
\usepackage{geometry}

\usepackage[authoryear]{natbib}
\usepackage{authblk}

\newtheorem{proposition}{Proposition}

\def\be{\begin{equation}}
\def\ee{\end{equation}}
\def\bea{\begin{eqnarray}}
\def\eea{\end{eqnarray}}

\def\X{\bm{X}}

\newcommand{\nocontentsline}[3]{}
\newcommand{\tocless}[2]{\bgroup\let\addcontentsline=\nocontentsline#1{#2}\egroup}

\DeclareMathSymbol{\leqslant}{\mathalpha}{AMSa}{"36} 
\DeclareMathSymbol{\geqslant}{\mathalpha}{AMSa}{"3E} 
\DeclareMathSymbol{\eset}{\mathalpha}{AMSb}{"3F}     
\renewcommand{\geq}{\;\geqslant\;}                   

\title{The Rise and Fall of Ideas' Popularity\thanks{We thank Matteo Bizzarri, Fabrizio Panebianco, Matteo Tanzi and Filippo Gusella for their helpful comments.
Paolo Pin acknowledges
funding from the Italian Ministry of Education Progetti di Rilevante Interesse
Nazionale (PRIN) grants 2017ELHNNJ, 2022389MRW and P20228SXNF. Piero Mazzarisi acknowledges funding from the Italian Ministry of University and Research under the
PRIN project {\it Realized Random Graphs: A New Econometric Methodology for the
Inference of Dynamic Networks} (grant agreement n. 2022MRSYB7). Alessio Muscillo acknowledges funding the University of Siena (F-New Frontiers 2023 grant) and from Universitas Mercatorum (FIN-RIC 2024 grant). 
The 
\href{https://github.com/alessiomuscillo/modeling_waves/blob/main/Code_for_Mazzisi_et_al_2024.ipynb}{code} used for the analysis is  \href{https://github.com/alessiomuscillo/modeling_waves}{available on GitHub}.
}}
\author[s]{Piero Mazzarisi}
\author[m]{Alessio Muscillo}
\author[s]{Claudio Pacati}
\author[s,b]{Paolo Pin}

\affil[s]{Department of Economics and Statistics, Universit\`a di Siena, Italy} 
\affil[m]{Department of Economics, Statistics and Business, Universitas Mercatorum, Italy}
\affil[b]{BIDSA, Universit\`a Bocconi, Milan, Italy}

\date{November 2024}

\begin{document}
\maketitle

\begin{abstract}
In the dynamic landscape of contemporary society, the popularity of ideas, opinions, and interests fluctuates rapidly. Traditional dynamical models in social sciences often fail to capture this inherent volatility, attributing changes to exogenous shocks rather than intrinsic features of the system. This paper introduces a novel, tractable model that simulates the natural rise and fall of ideas' popularity, offering a more accurate representation of real-world dynamics. Building upon the SIRS (Susceptible, Infectious, Recovered, Susceptible) epidemiological model, we incorporate a feedback mechanism that allows the recovery rate to vary dynamically based on the current state of the system. This modification reflects the cyclical nature of idea adoption and abandonment, driven by social saturation and renewed interest. Our model successfully captures the rapid and recurrent shifts in popularity, providing valuable insights into the mechanisms behind these fluctuations. This approach offers a robust framework for studying the diffusion dynamics of popular ideas, with potential applications across various fields such as marketing, technology adoption, and political movements.
\end{abstract}


\bigskip

{\bf Jel Classification codes}: {\bf C61} Dynamic Analysis -- {\bf D83}	Learning, Communication.

\section{Introduction}

Understanding the dynamics of idea diffusion and popularity cycles is crucial in numerous contexts, from marketing to political movements. Yet, existing models often fall short of capturing intrinsic fluctuations, relying on exogenous shocks to explain observed volatility. In this paper, we present a novel deterministic and tractable model that generates oscillating behavior—a feature not commonly addressed in traditional frameworks.

Our approach builds on the SIRS (Susceptible, Infectious, Recovered, Susceptible) model, a foundational framework in epidemiology used to describe the spread of infectious diseases. By reinterpreting the SIRS states in the context of idea diffusion, we model how individuals adopt, abandon, and potentially re-engage with ideas over time. Crucially, we extend the standard SIRS model with a feedback mechanism that dynamically adjusts the recovery rate based on the composition of the population. This simple yet powerful modification disrupts steady-state convergence, resulting in stable limit cycles that reflect real-world scenarios where the popularity of opinions and ideas oscillates over time.

Unlike traditional applications of the SIRS model, which focus on disease transmission in populations, our adaptation is uniquely suited to studying contexts where there is no “true state of the world,” and individual opinions evolve independently of objective realities. Applications of our framework extend to situations where opinions shape collective dynamics without directly influencing material outcomes, such as online trends, cultural phenomena, and speculative markets.
By focusing on the intrinsic feedback mechanisms driving cyclical behavior, our model offers a fresh perspective on the interplay between interest saturation and influencing enthusiasm. This approach complements and advances the current literature, providing a tool for understanding endogenous volatility in opinion dynamics which is simple and analytically tractable, yet robust.

The SIRS model, traditionally used in epidemiology to describe the spread of infectious diseases, can be adapted to study the diffusion of popular ideas within a population. Below is an interpretation of the Susceptible, Infectious, and Recovered states in the context of idea diffusion:

\begin{itemize}
    \item[(S):] \textbf{Susceptible} -- Individuals in this state are not currently aware of or influenced by the popular idea. However, they are open to adopting the idea if exposed to it. These individuals can be considered potential followers or adopters.
    \item[(I):] \textbf{Infectious} -- Individuals in this state have adopted the popular idea and are actively promoting it to others. They influence those in the susceptible state, encouraging them to adopt the idea. These individuals can be likened to enthusiastic supporters or advocates of the idea.
    \item[(R):] \textbf{Recovered} -- Individuals in this state were previously influenced by the idea but are no longer actively promoting it. They might have lost interest, become disillusioned, or moved on to other ideas. Over time, these individuals can return to the susceptible state, making them open to adopting the idea again or being influenced by new ideas.
    \item[(S):] \textbf{Susceptible again} -- After a period of being recovered, individuals return to the susceptible state. This reflects the cyclical nature of idea adoption, where individuals can re-engage with the same idea or become receptive to new ones over time.
\end{itemize}

The standard SIRS model is not ergodic and converges asymptotically to a steady state. However, in reality, we observe that the popularity of ideas and fads is often subject to significant fluctuations. For example, social media trends can spike rapidly, reaching wide adoption in a short period, only to diminish just as quickly as public interest wanes \citep{kwak2010twitter}. This phenomenon is not only observed in social media  but also in fashion \citep{aspers2013sociology}, technology adoption \citep{rogers2014diffusion}, political movements \citep{tufekci2014social}, and other areas where ideas compete for public attention. The dynamic nature of idea popularity underscores the need for models that can capture these rapid shifts and cyclical patterns. 

So, contrary to the standard model, we study a case where the parameter governing the transition from the infected to the recovered state is not constant. Instead, it depends linearly on both the number of infected individuals (positively) and the number of susceptible individuals (negatively). This interpretation allows us to capture the dynamic and recurring nature of how ideas spread and fade within a population, accounting for the repeated cycles of interest and disinterest that characterize popular ideas.

Positive dependence on the number of infected individuals implies that the more people are actively promoting the idea (infected), the faster individuals lose interest in the idea and transition to the recovered state. This could be due to oversaturation or the idea becoming less novel and interesting as more people adopt and promote it. In other words, widespread adoption and promotion can lead to a quicker decline in enthusiasm and a faster rate of disillusionment. We call this tendency \emph{interest saturation} (see \citealp{jetten2014deviance} for a survey of the literature in social psychology that describe deviant behaviors with respect to conformism).

Negative dependence on the number of susceptible individuals suggests that the more people are open to adopting the idea (susceptible), the slower individuals currently promoting the idea (infected) transition to the recovered state. This could be interpreted as a social dynamics effect where the presence of many potential adopters makes current promoters increase their focus, as if they had an objective function that increases in the number of successful transmissions. This is also in line with the assumptions at the basis of the economic analysis of \emph{cultural transmission}, started by \cite{bisin2001economics}. We call this tendency \emph{influencing enthusiasm}.

Remarkably, this model shows steady states that are fluctuating and can represent situations where the popularity of ideas is volatile, with fast cycles of adoption and abandonment influenced by social saturation and the shifting focus of the population's attention.

The rest of this paper is structured as follows. In Section \ref{sec:literature}, we review the relevant literature on idea diffusion, highlighting the key contributions of continuous- and discrete-time models, as well as agent-based and heterogeneous agent frameworks, and situate our model within this context. Section \ref{sec:model} introduces our modified SIRS model, detailing the endogenous feedback mechanisms of interest saturation and influencing enthusiasm that drive the oscillatory dynamics. We conduct a rigorous stability and bifurcation analysis to demonstrate how these mechanisms lead to periodic cycles in the system's behavior. Section \ref{sec:data} applies the model to empirical data, using Google Trends as a proxy for idea popularity to validate the predicted cyclical patterns. Finally, Section \ref{sec:conclusion} concludes with a discussion of the broader implications of our findings and potential avenues for future research.

The appendices provide technical details and proofs that support our theoretical findings. Appendix \ref{app:proofs} includes the full derivation of the mathematical results discussed in the stability and bifurcation analysis. Appendix \ref{app:simulations} offers some numerical computation results to illustrate the model’s behavior across a broader range of parameter values. 
Appendix \ref{app:data} describes the implementation details of the empirical validation, including the data pre--processing methods and the metric used to compare model predictions with Google Trends data.

\section{Relation to the Literature}
\label{sec:literature}

The diffusion of ideas, opinions, and interests within populations has been widely examined through various modeling approaches, many of which are inspired by epidemiological frameworks, heterogeneous agent models, and agent-based frameworks. These models shed light on how beliefs evolve and how individuals’ interests in certain ideas rise and fall over time as they interact with one another. 
In the present section, we review prominent approaches related to idea diffusion, particularly focusing on continuous- and discrete-time frameworks, heterogeneous agent models, and agent-based models. This provides a foundation for understanding the evolution of beliefs and expectations within complex social and economic systems, highlighting how our model diverges by allowing belief and interest changes to emerge from random, goal-independent interactions.

\subsubsection*{Continuous-Time Models}


Continuous-time models have been extensively employed to study the diffusion of ideas and behaviors, often drawing inspiration from epidemiological frameworks. Early contributions by \cite{young2006diffusion, young2009innovation} established foundational insights into how social influence and network structures affect the adoption of innovations over time. These works provided a dynamic perspective on the spread of ideas within populations, emphasizing the role of continuous interactions among individuals.
More recently, models have incorporated specific mechanisms to account for the cyclical nature of idea popularity. For instance, \cite{badr2021diffusion} and \cite{chen2020modeling} adapted a similar framework from epidemiology to model the rise and fall of ideas. These studies are based on numerical computation of the trajectories and also illustrate how individuals can repeatedly transition between adopting, abandoning, and reconsidering ideas, capturing the inherent volatility of social phenomena.

Our work introduces a modified SIRS framework that embeds a dynamic feedback mechanism. Unlike traditional continuous-time models, where transition rates between states are static, we allow the recovery rate to vary endogenously with the system’s state. This adjustment captures the interplay between social saturation—where widespread adoption accelerates disinterest—and influencing enthusiasm, where the presence of many potential adopters prolongs the efforts of active promoters. By introducing this feedback loop, our model departs from steady-state predictions and generates oscillatory dynamics, reflecting real-world patterns of idea diffusion.

This advancement underscores the versatility of continuous-time models in representing complex social dynamics. By incorporating endogenous mechanisms, we extend their applicability to contexts where rapid and recurrent fluctuations dominate, offering a robust framework for analyzing opinion dynamics in diverse fields, from cultural trends to digital marketing.

\subsubsection*{Discrete-Time Models}

While continuous-time models provide valuable insights, most models in this domain operate in discrete time, reflecting the periodic nature of data collection and decision-making in real-world scenarios. These models often involve a discrete number of agents and typically converge to a steady-state equilibrium, where the spread of the idea stabilizes. However, several studies highlight exceptions to this pattern. \cite{acemouglu2013opinion} demonstrated that the presence of stubborn agents—individuals resistant to changing their opinions—can prevent convergence, leading to persistent fluctuations in opinion dynamics. Similarly, the recent work by \cite{danenberg2024endogenous} introduces a stochastic component into the model, reflecting the random nature of attention and information spread in digital environments. This addition helps to capture the unpredictable nature of idea diffusion in the context of modern social media platforms.


A concept akin to the ``recovered" or ``uninterested" individuals in the SIRS model is present in the influence campaigns studied by \cite{sadler2023influence}. In this model, individuals who lose interest in a particular idea may still be susceptible to future influence, particularly as new or modified ideas emerge. This cyclical susceptibility is a key feature in understanding how ideas wax and wane in popularity, similar to the patterns observed in our extended SIRS framework.

Other simulation-based studies provide insights into the dynamics of idea diffusion in more complex settings. \cite{hethcote1981nonlinear} explored nonlinear oscillations in epidemic models, demonstrating how small changes in parameters can lead to significant fluctuations in outcomes. \cite{khalifi2022extending} extended a traditional  model in discrete time by allowing for gradual waning of immunity, which parallels the gradual loss of interest in an idea. These studies, though focused primarily on disease dynamics, offer valuable analogies for understanding the nonlinear and often unpredictable patterns of idea spread in social systems.

\subsubsection*{Agent-Based and Heterogeneous Agent Models for Belief Dynamics}

Agent-based models (ABMs) and heterogeneous agent models (HAMs) are powerful tools for studying belief dynamics in populations, offering insights into how individual-level interactions shape aggregate outcomes. These frameworks typically rely on structured rules or optimization-driven behaviors, often assuming the presence of an underlying ``true state of the world'' to which agents respond. While this assumption is useful in contexts such as financial markets or opinion formation under clear external signals, it becomes less applicable in scenarios where beliefs and interests evolve independently of an objective reality.

Our approach diverges from this literature by removing the need for a true state of the world. Instead, we model belief and interest dynamics as emerging purely from decentralized and random interactions among agents, without anchoring to external truths. This abstraction allows us to focus on endogenous processes, such as the cyclical patterns of interest saturation and influencing enthusiasm, which drive opinion fluctuations and the rise and fall of ideas. By doing so, our model provides a distinct lens to study belief dynamics in settings where subjective perceptions and social interactions dominate, offering a novel contribution to the literature.

Within the ABM framework, HAMs play a specific role by categorizing agents based on different expectations and strategies, particularly in financial markets. Foundational works by \cite{day1990bulls,chiarella1992dynamics,brock1998heterogeneous} and \cite{lux1995herd} have shown how diverse beliefs can drive complex market dynamics, such as price fluctuations and trend formation (see also \citealt{lebaron2008modeling,hommes2021behavioral}). Typically, HAMs assume agents have distinct objectives or payoff functions, such as aligning with market fundamentals or following trends.

Recent advances in ABM validation have introduced machine learning and Bayesian estimation techniques to address the challenges of calibrating models with high agent heterogeneity. For example, \cite{lamperti2018agent} use machine learning surrogates for calibration, while \cite{gatti2020rising} propose Bayesian methods for parameter estimation and forecasting. These methods enhance the robustness of ABMs by improving the reliability of agent-level decision processes, particularly in macroeconomic settings \citep{dosi2019more,monti2023learning}.

An important aspect of these models is the role of endogenous fluctuations, where cycles in belief dynamics emerge naturally from the interactions between agents. This approach has been advanced in a recent work by \cite{gusella2024endogenous}, who use a state-space framework to capture the endogenous cycles generated by heterogeneous expectations. This concept of endogenous fluctuations is fundamental to our model as well.

In contrast to traditional HAMs and ABMs, our model abstracts from specific payoff functions or objective-driven behaviors. Rather than categorizing agents by strategy or using structured decision rules, we capture belief evolution through random, decentralized exchanges, allowing interest levels to fluctuate organically. This approach models the spontaneous evolution of beliefs without relying on goal-oriented updates or optimization processes, providing a unique perspective within the ABM and HAM literature.

Although agent-based and heterogeneous agent models have provided valuable insights into belief dynamics, none of them share a technical similarity with our approach. These models typically employ agent-level simulations, optimization-driven strategies, or complex decision-making processes to explain aggregate dynamics. In contrast, our model is a simple system of differential equations, offering a tractable and analytically solvable framework for studying belief evolution. By focusing on continuous and deterministic dynamics rather than relying on agent-specific rules or stochastic simulations, our model provides a distinct methodological perspective that has not been explored in this literature. This technical simplicity, combined with its ability to capture cyclical patterns, sets our work apart from the agent-based and heterogeneous agent modeling traditions.

\section{The model}
\label{sec:model}

The SIRS model is represented by a system of ordinary differential equations (ODEs) that describe the rates of change for the three states over time:

\begin{itemize}
\item[(S):] {\bf Susceptible} - The rate of change in the susceptible population ($dS/dt$) is positively influenced by the rate at which individuals become {\it open} to adopt the idea and become susceptible ($\xi R$ with $\xi>0$), and is decreased by the rate at which susceptible individuals {\it adopt} actually it ($-\beta SI$ with $\beta>0$).
\begin{equation}\label{eq:S}
\frac{dS}{dt} = -\beta SI + \xi R,
\end{equation}
\item[(I):] {\bf Infectious} - The rate of change in the population ``infected'' by the idea ($dI/dt$) is increased by the rate at which susceptible individuals adopt the idea ($+\beta SI$) and decreased by the rate at which individuals that previously adopted the idea change their mind ($-\gamma I$ with $\gamma> 0$).
\begin{equation}\label{eq:I}
\frac{dI}{dt} = \beta SI - \gamma I,
\end{equation}
\item[(R):] {\bf Recovered} - The rate of change in the population that changed their mind $dR/dt$ is determined by the rate at which individuals infected by the idea change their mind ($\gamma I$) and the rate at which individuals become open to adopt the idea again ($-\xi R$).
\begin{equation}\label{eq:R}
\frac{dR}{dt} = \gamma I - \xi R
\end{equation}
\end{itemize}
The set of Equations (\ref{eq:S}), (\ref{eq:I}), and (\ref{eq:R}) define a dynamical system, which is characterized by an invariant of the dynamics, namely the total number of individuals $S+I+R=1$ (normalized to one for convenience), allowing to reduce by one the number of equations by constraining one dependent variable, e.g. $R=1-S-I$. In these equations, (i) $\beta$ is the transmission rate, representing the likelihood that the idea spreads from ``infected'' population to individuals open to change their mind; (ii) $\xi$ is the rate at which individuals who were closed to adopt the new idea becomes open to it; finally, (iii) $\gamma$ is the `recovery' rate at which individuals change their mind and abandon the idea previously adopted. The intuition on these transition rates from one state to the other mediated by constant parameters $\beta$, $\xi$, and $\gamma$ is crucial: the period of stay in a given state is the inverse of the rate. For example, all individuals in $I$ at time $t$ transition to $R$ after a period $1/\gamma$.

The standard SIRS model predicts a steady state for the long-run dynamics, see, e.g., \cite{anderson1991infectious}.  However, many real-world epidemiological examples exhibit dynamics that deviate significantly from the predictions of the SIRS model, displaying oscillations, mean reversion, and often periodic evolutions. Attempts to include such patterns in simple epidemic models in continuous time date back to 70's and proposed solutions revolve around a few key intuitions: (i) exogenous time-varying patterns, possibly seasonal, for the transmission rates as in \cite{hethcote1973asymptotic}; (ii) delay effects for the transition rates resulting in integro-differential equations as in \cite{cooke1976periodicity}; (iii) more complex interactions like a nonlinear incidence for the transmission rate ($\beta S^pI^q$) as in \cite{hethcote1989epidemiological}.

We model the coupled mechanism of {\it interest saturation} and {\it influencing enthusiasm} within the general context of the SIRS framework in a straightforward manner. Specifically, we make the recovery rate $\gamma$ a dependent variable, denoted as $\Gamma(t)$, and we model a linear dependence of $\log\Gamma(t)$ on the states $I(t)$ and $S(t)$ as\footnote{Modeling $\log\Gamma$ in the place of $\Gamma$ is convenient to ensure that the recovery rate remains positive.}
\begin{equation}\label{eq:G}
\frac{d\Gamma}{dt} = \Gamma(\alpha I-\delta S),
\end{equation}
with $\alpha,\delta> 0$. This approach describes a recovery transition process, where the transition from infected ($I$) to recovered ($R$) happens more quickly when the number of infected individuals is high (indicating that many infected individuals lose interest in a commonly adopted idea) and more slowly when the number of susceptible individuals is high (indicating that a few infected individuals retain the idea longer to promote its spread). In other words, the state $R$ can be considered as a reservoir that introduces a delay in the transition from $I$ to $S$, while the recovery transition process acts as a feedback mechanism, which favors the less represented compartment between $I$ and $S$ when the other one has increased its numbers.

We demonstrate that this straightforward feedback mechanism disrupts the steady-state equilibrium, breaking the stability of the fixed-point dynamics. In a restricted version of the model ($\alpha=\delta$), we prove this happens through a Hopf bifurcation, which leads to the emergence of stable limit cycles.
Furthermore, from a methodological standpoint, it is noteworthy that cycles can arise endogenously within one of the simplest systems of interactions. This occurs without the need for external drivers or complex delay effects, providing a very simple framework for studying periodic behavior in opinion dynamics.

\subsection{Stability analysis}
The model reads as the following system of ODEs
\begin{align}
\frac{dS}{dt} & = -\beta SI + \xi (1-S-I),\label{eq:SIRSfs}\\
\frac{dI}{dt} & = \beta SI - \Gamma I,\label{eq:SIRSfi}\\
\frac{d\Gamma}{dt} & = \Gamma(\alpha I-\delta S)\label{eq:SIRSfg},
\end{align}
where $\X\equiv(S,I,\Gamma)\in [0,1]^2\times \mathbb{R}_{>0}$ and $\beta,\xi,\alpha,\delta\in\mathbb{R}_{>0}$, with the flow operator defined by
\begin{equation}\label{eq:flow}
    \Phi(\X) = \begin{pmatrix}
        -\beta SI + \xi(1-S-I)\\
        \beta SI - \Gamma I \\
        \Gamma (\alpha I - \delta S)
    \end{pmatrix}.
\end{equation}
The solutions of $\Phi(\X)={\bm0}$ define the fixed points of the system. In the domain of the state vector $\X$, there exists only one solution\footnote{There are two further solutions $X^*_{10}=(1,0,0)$ and $X^*_{01}=(0,1,0)$ of $\Phi(X)=0$, but outside the domain (indeed at the boundary) because of $\Gamma^*_{10}=\Gamma^*_{01}=0$.}
\begin{equation}\label{eq:fp}
\X^*= \begin{cases}
    \displaystyle S^* = \frac{-(\delta+\alpha)+\sqrt{\left(\alpha +\delta\right)^2+4\alpha\delta\frac{\beta}{\xi}}}{2\delta\frac{\beta}{\xi}},\\
    \displaystyle I^* = \frac{\delta}{\alpha}S^*,\\
    \displaystyle \Gamma^* = \beta S^*.
    \end{cases}
\end{equation}
The stability of $\X^*$ can be studied in perturbation analysis via the Jacobian of the flow operator, see, e.g., \cite{guckenheimer2013nonlinear}.

\begin{proposition}
\label{unstable}
Solution $\X^*$ for ODEs~\eqref{eq:SIRSfs}--\eqref{eq:SIRSfg} is locally unstable if and only if:
\begin{align}
 \left\{
  \begin{aligned}
  \alpha& > \beta + \xi \\
  \delta& \geq   \frac{\alpha^2\xi}{(\alpha-\beta)(\alpha-\beta-\xi)}.
  \end{aligned}\label{eq:RH2}
 \right.
\end{align}.
\end{proposition}
\begin{proof}
See Appendix \ref{sec:appendix_proof_prop_1}.
\end{proof}

\begin{figure}[t]
\centering
\includegraphics[width=0.47\textwidth]{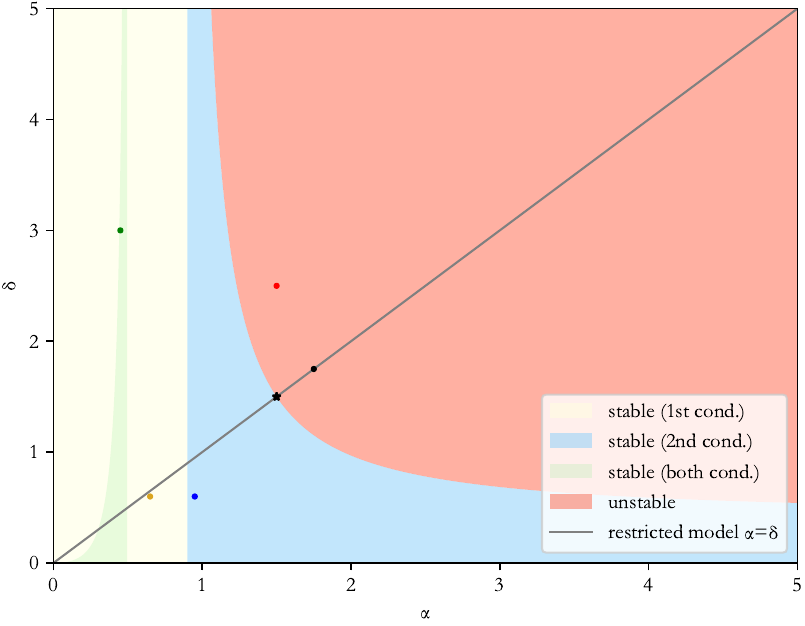}\qquad
\includegraphics[width=0.47\textwidth]{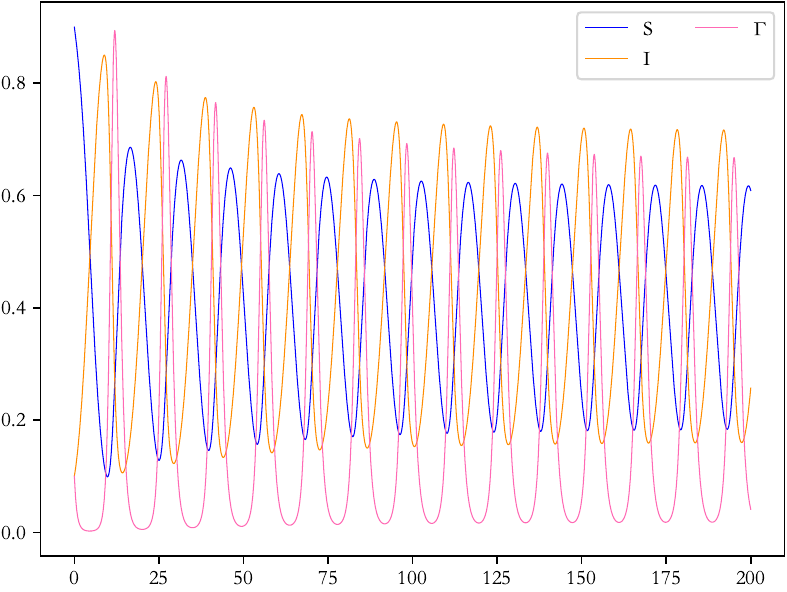}
\caption{Left: layout in the $(\alpha,\delta)$ parameter space, for $\beta=0.5$ and $\xi=0.4$, of the region of instability (red), and of stability due to the first condition (yellow), to the second (blue) or to both conditions (green) of Eq.~\eqref{eq:RH}. The five dots represent the parameters used to produce the graph in the right panel (black dot), and  Figures~\ref{fig:stab2} (yellow dot), \ref{fig:stab3} (blue dot), \ref{fig:stab1} (green dot), \ref{fig:nonstab} (red dot) in Appendix~\ref{app:simulations}.
The black star highlights the Hopf bifurcation point of the restricted model with $\alpha=\delta$ (see Section~\ref{sec:bifhopf}). Right: Solution of the model for $\beta=0.5$, $\xi=0.4$, and $\alpha=\delta=1.75$ (black dot on the left graph), with $S(0)=0.9$, $I(0)=\Gamma(0)=0.1$.}\label{fig:stabregions}
\end{figure}

When $\X^*$ is an unstable fixed point, also known as a repeller, small perturbations away from the fixed point lead to trajectories that diverge from it over time. This instability means that the fixed point cannot act as an attractor for the system; instead, it causes nearby trajectories to be expelled in its vicinity. As a result, the system tends to evolve away from the unstable fixed point, potentially leading to either periodic, complex, or divergent dynamics. This can manifest in various ways, such as chaotic behavior, bifurcations, or transitions to other regions of the state space, significantly influencing the overall long-term behavior of the system. Figure~\ref{fig:stabregions} shows an example decomposition of the parameter space into stable and unstable regions. Extensive numerical computaions of the model in the region of instability support the conclusion that stable limit cycles appear, leading to a periodic evolution of the system.

We demonstrate below that a restricted version ($\alpha=\delta$) of the model undergoes a Hopf bifurcation at the crossing point to the unstable region, under a condition on parameters $\beta$ and $\xi$.

\subsection{Bifurcation analysis}\label{sec:bifhopf}
A Hopf bifurcation is a critical phenomenon in dynamical systems where a fixed point of a system undergoes a qualitative change as a parameter is varied, leading to the emergence of a periodic orbit. Specifically, it occurs when a pair of complex conjugate eigenvalues of the Jacobian matrix of the flow operator crosses the imaginary axis from the left half-plane to the right half-plane while all the other eigenvalues have negative real part, see, e.g., \cite{marsden2012hopf}. This transition indicates a shift from stability to instability of the fixed point, resulting in oscillatory behavior.

\begin{proposition}
\label{hopf}
Solution $\X^*$ of ODEs~\eqref{eq:SIRSfs}--\eqref{eq:SIRSfg} with $\alpha=\delta$ is a Hopf bifurcation at $$
\alpha=\beta+\xi+\sqrt{\xi(\beta+\xi)}
$$ 
if
$$
\beta>\frac{11}{25}\xi.
$$
\end{proposition}
\begin{proof}
See Appendix \ref{sec:appendix_proof_prop_2}.
\end{proof}
This sufficient condition provides a rigorous example of the model displaying periodic orbits when the fixed-point equilibrium is broken, confirming the intuition that the combined mechanism of interest saturation and influencing enthusiasm could give rise to repeated cycles of interest and disinterest in opinion dynamics as described within the SIRS modeling framework. For example Figure~\ref{fig:stabregions} shows the paths followed by the three state variables of the restricted model in the region of the instability, displaying their periodic orbits.

Indeed, further insights into the behavior of our model under a wide range of parameter values are provided in Appendix \ref{app:simulations}. There, we present an extensive set of numerical computation of trajectories, illustrating how the dynamics of the system change with variations in key parameters such as transmission rate, recovery rate, and feedback strength. These computations  complement the theoretical results discussed here, offering a deeper understanding of the model’s oscillatory patterns and its robustness across different scenarios.

\section{Empirical application}
\label{sec:data}

Digitalization allows us to track ideas' popularity in many ways nowadays, from low-level measures based on granular social network data to high-level indicators aggregating all the information.
\href{https://trends.google.com/trends/}{Google Trends} represent aggregated indicators of the popularity of {\it research queries} in Google Search across various regions and languages, measured as the relative weekly search volumes normalized to the highest value in the investigated period, which is set to 100.
\begin{figure}[th]
    \centering
    \includegraphics[width=\textwidth]{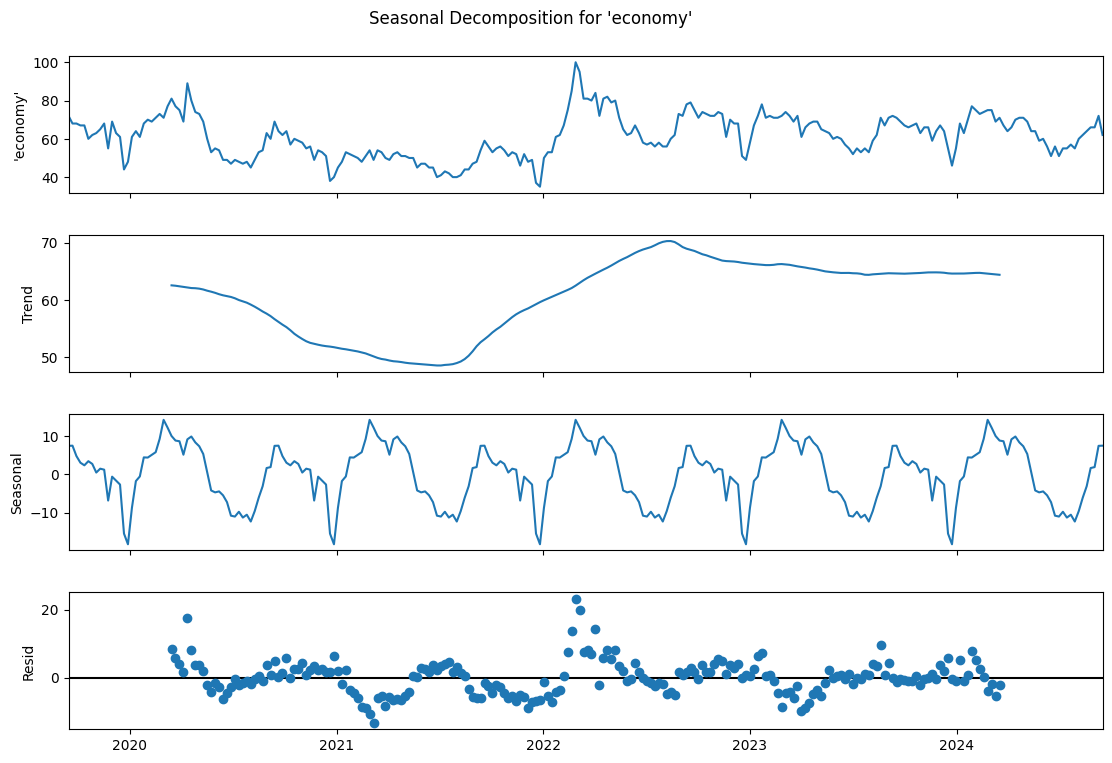}
    \caption{Seasonal decomposition for the word `{\it economy}' searched in Google Trends. After removing the trend and annual seasonality, we are left with the residuals (bottom panel). See Appendix~\ref{app:seasonal_decomposition}
    for additional information.}
    \label{fig:economy}
\end{figure}
We show the research query for `{\it economy}' in a 5-year period from September 15th, 2019 to September 15th, 2024 in the top panel of Figure~\ref{fig:economy}. Observed search volumes can depend on several effects, either exogenous or endogenous.

Media coverage of a particular topic can generate significant activity in Google searches, leading to trends in the search volume over time. Seasonal factors naturally create periodic patterns in Google searches. For example, queries like `{\it swimsuit}' or `{\it skiing}' are largely influenced by the time of year. Generally, any event, whether regional or global, can capture people's attention, causing random shifts in search volume dynamics. Meanwhile, conversations, discussions, and other forms of communication can also draw people's attention to specific topics. As such, social interactions can drive search volumes over time. There may also exist many other social mechanisms, all of them contributing to the aggregated dynamics of Google trends.

This section aims to assess whether Google Trends, as a proxy for the popularity of ideas, exhibits dynamical patterns that can be consistently explained by the endogenous effects outlined in our model. After accounting for known influences like trending and seasonality, endogenous effects should manifest as cyclical patterns instead of purely random evolution.



We collected a dataset of the 1000 most searched queries in the United States in September 2024 from \href{https://www.google.com/url?q=https%3A%2F%2Fdataforseo.com%2Ffree-seo-stats%2Ftop-1000-keywords}{Data For SEO} website.\footnote{%
The preprocessing of Google Trends data, including detrending and standardization methods, is discussed in Appendix \ref{app:data}. These steps ensure that the observed patterns are suitable for comparison with our model’s predictions.}
For each word, the time series of search volume observed on any week from September 15th, 2019 to September 15th, 2024 is a proxy of the word's popularity over time, providing insight on how interest in each topic fluctuates. For a genuine assessment, we first remove both trend and annual seasonality: (i) trend is estimated as a local average of observations over a rolling time window of 52 weeks, while (ii) annual seasonality is estimated starting from the detrended time series of observed search volume as the average of the same weeks over every year. An example of the two patterns for the word `{\it economy}' is shown in the middle panels of Figure~\ref{fig:economy}. After removing trend and seasonality, the time series of residuals (see the bottom panel of Figure~\ref{fig:economy}) is standardized by rescaling data in a range $[-1,1]$ and compared with two reference processes: the proposed model of endogenous oscillations – a system simulating regular, internally driven cycles of `infected' - and a random walk - a model describing the contagion of an idea as a purely random variation of `infected' at any time. Here, we connect the concept of `infected' with the popularity of an idea, based on the intuition that the larger the number of `infected,' the larger the observed search volume of the corresponding word. For a fair comparison, the same standardization scheme applied to the time series of residuals is also applied to the simulations of both reference processes. 

The comparison uses Dynamic Time Warping (DTW) as distance, see, e.g., \cite{hastie1991model}. DTW is the generalization of the Euclidean distance between two time series, allowing for non-linear stretching or shrinking of the time axis to find the optimal alignment between the two. DTW is particularly useful when comparing time series that may not be perfectly synchronized but display similar patterns.
We use dynamic DTW to measure the similarity between the observed time series and the trajectories generated by our model. A comprehensive explanation of the DTW metric and its implementation is provided in Appendix \ref{app:data}.


We first measure the DTW distance between the time series of standardized residuals and our model with the best value for $\beta$ (i.e. the transmission rate) in terms of DWT.\footnote{A detailed explanation of the implementation is in Appendix \ref{app:data} and the code is available at \href{https://github.com/alessiomuscillo/modeling_waves/blob/main/Code_for_Mazzisi_et_al_2024.ipynb}{this link}.} Then, the DTW distance is computed for 500 simulations of a random walk, whose average is finally compared to our model. 
Figure~\ref{fig:distances} shows the scatter plot between the DWT distances associated with our model and the random walks for the 1000 most searched queries. The time series of residuals systematically show a closer match to the dynamic behavior predicted by our model. To ensure robustness, a circle in the scatter plot is filled with blue when the DTW distance associated with our model is smaller than the first quartile of the empirical distribution built with 500 simulations of the random walk.
Our analysis shows that the endogenous cyclical effects capture the underlying patterns of Google searches statistically better than what would be expected from random fluctuations in the 67\% of cases.

\begin{figure}[t]
    \centering
    \includegraphics[width=.65\textwidth]{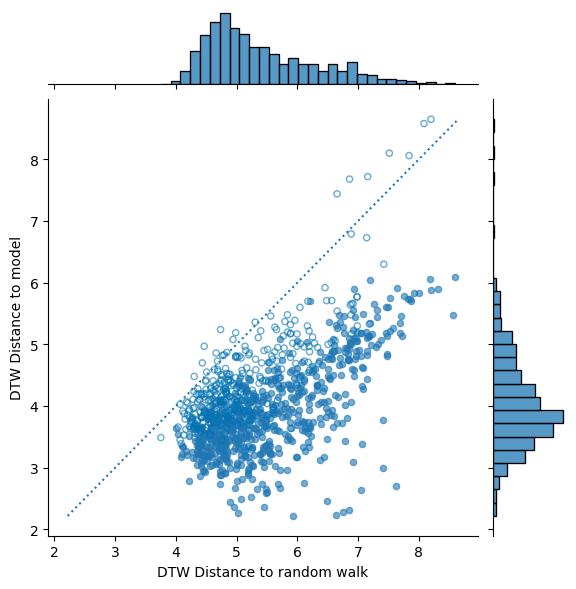}
    \caption{Scatter plot showing the relationship between the Dynamic Time Warping (DTW) distance of residuals to a random walk (x-axis) and the DTW distance to the model (y-axis). (See 
    Appendix~\ref{app:distance_residuals} for additional information.) Each circle represents a word from the dataset. Marginal histograms indicate the distribution of DTW distances for both reference processes. The dotted diagonal line represents the case when the average DTW distance to the random walk equals the DTW distance to the model. A circle below the diagonal suggests that the residuals are more similar to our model than to the random walk. The circle is filled with blue when the DTW distance associated with the model is smaller than the first quartile of the empirical distribution built with 500 simulations of the random walk, indicating that the model better captures the underlying oscillations, in a statistical sense.}
    \label{fig:distances}
\end{figure}

These results provide evidence that the temporal dynamics of online search behavior are not purely random but exhibit structured oscillations that can be explained, at least in part, by endogenous processes. Our model, therefore, offers a meaningful framework for understanding how certain search queries experience periodic surges in interest, reflecting broader patterns in human behavior.

\clearpage

\section{Conclusion}
\label{sec:conclusion}

Our analysis demonstrates the existence of complex dynamics, including cycles and potential bifurcations, providing insights into the underlying mechanisms that govern the rise and fall of ideas. This framework can be extended to explore additional feedback mechanisms and their broader implications in the diffusion of innovations and cultural trends.
The feedback mechanism introduced in our modified SIRS model captures the dynamic interaction between interest saturation and influencing enthusiasm, leading to fluctuating steady states. This cyclical behavior mirrors real-world observations where ideas periodically gain and lose popularity.

Moreover, our study highlights the importance of considering endogenous factors in modeling idea diffusion. The rate at which individuals lose interest (interest saturation) and the influence exerted by current promoters (influencing enthusiasm) are both crucial in understanding how ideas spread and fade. By incorporating these factors, our model provides a more comprehensive view of the diffusion process, which can be applied to various fields, including marketing, public health campaigns, and the spread of innovations.

Future research could expand on this framework by incorporating stochastic elements to account for random external influences and by exploring the impact of network structures on the diffusion dynamics. Additionally, empirical validation of the model using real-world data on idea diffusion could further enhance its applicability and robustness.

In conclusion, our modified SIRS model with feedback provides a valuable tool for understanding the complex and dynamic nature of idea popularity. It offers a theoretical foundation that can be built upon to explore various aspects of social influence and cultural evolution, enabling the development of more effective strategies in managing and promoting ideas.


\bibliographystyle{ecca}
\bibliography{biblio}

\appendix

\section{Proofs}
\label{app:proofs}

\subsection{Proof of Proposition \ref{unstable}}
\label{sec:appendix_proof_prop_1}

\begin{proof}
First, we compute the eigenvalues of the first-order derivatives of $\Phi$ (the Jacobian).

At $\X=(S,I,\Gamma)$, the Jacobian of $\Phi$ is
\[
J(\bm x) = \begin{pmatrix}
                -(\beta I+\xi)  &-(\beta S + \xi) &0\\
                \beta I         &\beta S-\Gamma   &-I\\
                -\delta\Gamma   &\alpha\Gamma     &\alpha I-\delta S
            \end{pmatrix}
            \enspace,
\]
having characteristic polynomial
\begin{align}
p(\lambda) &= \det[J(\X)-\lambda \mathbb{I}]  \notag\\
           &=-\lambda^3\notag\\
           &\quad- [(\beta-\alpha) I+(\delta-\beta) S+\Gamma +\xi]\lambda^2 \notag\\
           &\quad- (\alpha\Gamma+\beta^2S+\beta\xi)I\lambda \notag\\
           &\quad- [\beta(2\delta\Gamma-\beta\delta S+\alpha\beta I)SI+\beta\xi(\alpha I-\delta S)S+\delta\xi(S+I)\Gamma]\label{eq:charpolgen}\enspace.
\end{align}
Recall that $p(\lambda)$ is said to be \emph{stable} if all roots have negative real part.
Ignoring uninteresting cases, assume all parameters to be strictly positive and $\X=(I,S,\Gamma)\in[0,1]^2\times \mathbb{R}_{>0}$. Under these assumptions, the odd-order coefficients of $p(\lambda)$ are strictly negative. 
The Routh--Hurwitz stability criterion \cite[sec.~11.4]{rahman2002analytic} for $p(\lambda)$ becomes then:
\begin{gather}
\left\{
\begin{aligned}
(\beta-\alpha) I+(\delta-\beta) S+\Gamma +\xi &>0\\
\beta(2\delta\Gamma-\beta\delta S+\alpha\beta I)SI+\beta\xi(\alpha S-\delta S)S+\delta\xi(S+I)\gamma&>0\\
(\alpha\Gamma+\beta^2S+\beta\xi)[(\beta-\alpha) I+(\delta-\beta) S+\Gamma +\xi]
-\beta(2\delta\Gamma-\beta\delta S+\alpha\beta I)SI\\
-\beta\xi(\alpha I-\delta S)S-\delta\xi(S+I)\Gamma&>0
\end{aligned}
\right.\label{eq:RHgen}
\end{gather}
As model~\eqref{eq:SIRSfs}--\eqref{eq:SIRSfg} is non-linear, condition~\eqref{eq:RHgen} is a necessary and sufficient condition for \emph{local} stability at $\X$ of the model.

At the fixed point $\X^*=(I^*,S^*,\Gamma^*)$ in Eq. (\ref{eq:fp}), the characteristic polynomial~\eqref{eq:charpolgen} simplifies as
\begin{align}
p(\lambda) &= -\lambda^3 - (\beta I^*+\xi)\lambda^2 - (\alpha\Gamma^*+\beta^2S^*+\beta\xi)I^*\lambda - (\beta\delta S^*+\delta\xi+\alpha\beta I^*+\alpha\xi)\Gamma^* I^*
\enspace,\label{eq:charpoly}
\end{align}
and has strictly negative coefficients, hence the first two inequalities of~\eqref{eq:RHgen}
are satisfied. Therefore, the local stability condition at $\X^*$ is the third inequality only, which with some algebra can be stated as
\begin{align}
  \alpha&\le\beta+\xi
  &&\text{or}
  &\delta&<\frac{\alpha^2\xi}{(\alpha-\beta)(\alpha-\beta-\xi)}\enspace.\label{eq:RH}
\end{align}
Therefore, its complement, i.e.\ the statement, provides the local instability condition.

This concludes the proof.
\end{proof}

\subsection{Proof of Proposition \ref{hopf}}
\label{sec:appendix_proof_prop_2}

\begin{proof}
When $\alpha = \delta$, the fixed point $\X^*$ in Eq.~(\ref{eq:fp}) reads as
\begin{equation}\label{eq:fprestricted}
\X^*= \begin{cases}
    \displaystyle S^* = \frac{\left(-1+\sqrt{1+\frac{\beta}{\xi}}\right)}{\frac{\beta}{\xi}} = \frac{1}{1+\sqrt{\frac{\beta+\xi}{\xi}}},\\
    \displaystyle I^* = S^*,\\
    \displaystyle \Gamma^* = \beta S^*.
    \end{cases}
\end{equation}
and the characteristic polynomial in Eq.~(\ref{eq:charpolgen}) at $\X^*$ simplifies as
\begin{equation}\label{eq:charpolyre}
\begin{split}
p(\lambda) &= -\lambda^3 - (\beta S^*+\xi)\lambda^2 - (\alpha\beta S^*+\beta^2S^*+\beta\xi)S^*\lambda - [2\alpha\beta^2(S^*)^3+2\alpha\beta\xi(S^*)^2]\\
& =a\lambda^3+b\lambda^2+c\lambda+d
\end{split}
\end{equation}
where
\begin{align*}
a &= -1,
&b &= -\sqrt{\xi(\beta+\xi)},
&c  &= -\frac{\alpha+\beta}{\beta}(\xi+b)^2 +\xi(\xi+b),
&d  &=  2\frac\alpha\beta b(\xi+b)^2.
\end{align*}
The solution of the cubic equation $a\lambda^3+b\lambda^2+c\lambda+d=0$ is well-known in Algebra; see, e.g., the pioneering work by \cite{cardano2007rules}. First, a change of variable is applied, i.e. $\lambda=y-\frac{b}{3a}$. Then, a reduced form $y^3+py+q=0$ is obtained, with (after some algebra)
\begin{align*}
p &= \frac{c}{a}-\frac{b^2}{3a^2}
  =\frac{\alpha\xi +(\alpha+\frac23\beta)(\beta+\xi) +(2\alpha+\beta) b}{\beta}\xi,\\
q &= \frac{d}{a}-\frac{bc}{3a^2}+\frac{2b^3}{27a^3}
=\left[-\frac53\frac\alpha\beta (\xi+b)^2 + \frac13\xi b + \frac7{27}b^2\right]b.
\end{align*}
The type of solution of Eq.~(\ref{eq:charpolyre}) depends on
\begin{equation}
    \Delta = \frac{q^2}{4}+\frac{p^3}{27};
\end{equation}
in particular, the cubic equation has one real root and two complex conjugate roots if and only if $\Delta>0$.

A sufficient condition for $\Delta>0$ is $p>0$. The condition $p>0$ is equivalent to
$$
\alpha\xi+(\alpha+\frac23\beta)(\beta+\xi) -(2\alpha+\beta)\sqrt{\xi(\beta+\xi)}>0,
$$
which becomes after some algebra
$$
\alpha\left[\left(\sqrt{\xi}-\sqrt{\beta+\xi}\right)^2\right]>\left(\sqrt{\xi}-\frac23\sqrt{\beta+\xi}\right)\beta\sqrt{\beta+\xi}.
$$
As $\beta,\xi>0$, the left-hand side of this last inequality is positive. Therefore, the sufficient condition for $\Delta>0$ can be stated as
\begin{equation}\label{eq:p>0}
    \alpha>\frac{\left(\sqrt{\xi}-\frac23\sqrt{\beta+\xi}\right)}{\left(\sqrt{\xi}-\sqrt{\beta+\xi}\right)^2}\beta\sqrt{\beta+\xi}.
\end{equation}
When Eq.~(\ref{eq:p>0}) is satisfied, the characteristic polynomial in Eq.~(\ref{eq:charpolyre}) has one real root and two complex conjugate roots. Moreover, the real root is always negative because of Descartes' rule of signs applied to Eq.~(\ref{eq:charpolyre}), but with a sign change for the variable ($\lambda\rightarrow -\lambda$).

In the domain of parameters, the crossing point to the unstable region is given by the second inequality in Eq.~(\ref{eq:RH2}) with equality sign for $\alpha=\delta$, provided that the first inequality is satisfied. It is
\begin{equation}\label{eq:cp}
    \alpha = \beta+\xi+\sqrt{\beta\xi+\xi^2}.
\end{equation}

It remains to prove that the crossing point from stable (all roots have negative real part) to unstable (the two complex conjugate roots have positive real part) regions is an interior point of the parameter subspace defined by Eq.~(\ref{eq:p>0}). In this case, Proposition~\ref{unstable} guarantees that the real part of the two complex conjugate roots is zero at the crossing point because of the continuity of the solution of the cubic equation.

The crossing point is an interior point of the open subspace defined by Eq.~(\ref{eq:p>0}) when
$$
\beta+\xi+\sqrt{\beta\xi+\xi^2}>\frac{\left(\sqrt{\xi}-\frac23\sqrt{\beta+\xi}\right)}{\left(\sqrt{\xi}-\sqrt{\beta+\xi}\right)^2}\beta\sqrt{\beta+\xi}.
$$
The last inequality can be stated as
\begin{equation}\label{eq:intpoint}
\beta\frac{5(\beta+\xi)-6\sqrt{\xi(\beta+\xi)}}{3(\beta+2\xi)-6\sqrt{\xi(\beta+\xi)}}>0.
\end{equation}
Eq.~(\ref{eq:intpoint}) is satisfied when both the numerator and denominator on the left-hand side are either positive or negative. As $\beta,\xi>0$, the latter condition must be discarded. The former one results in a system of inequalities whose solution is
$$
\beta>\frac{11}{25}\xi.
$$
This concludes the proof.
\end{proof}





\section{Numerical solutions of the model}
\label{app:simulations}

Remembering the interpretation of $\alpha$ and $\delta$ as interest saturation and influencing enthusiasm, the inequalities from Proposition~\ref{unstable} tell us that the system fluctuates if and only if they are both high enough. Figure~\ref{fig:stabregions} shows an example decomposition of the parameter space into stable and unstable regions.
\begin{figure}[!htb]
\centering
\includegraphics[width=0.7\textwidth]{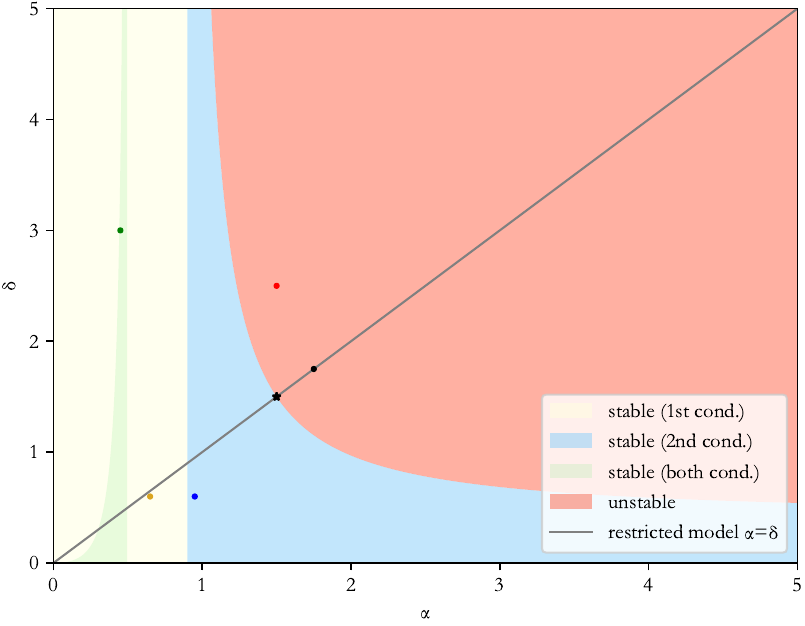}
\caption{Layout in the $(\alpha,\delta)$ parameter space, for $\beta=0.5$ and $\xi=0.4$, of the region of instability (red), and of stability due to the first condition (yellow), to the second (blue) or to both conditions (green). The five dots represent the parameters used to produce Figures~\ref{fig:stab2} (yellow dot), \ref{fig:stab3} (blue dot), \ref{fig:stab1} (green dot), \ref{fig:nonstab} (red dot), and~\ref{fig:nonstab2} (black dot). The black star highlights the Hopf bifurcation point of the restricted model with $\alpha=\delta$, the effects of crossing which are shown in Figures~\ref{fig:Hopf1} and~\ref{fig:Hopf2}.}\label{fig:stabregions}
\end{figure}

Figures~\ref{fig:stab2}, \ref{fig:stab3}, and~\ref{fig:stab1} show examples of the model dynamics when the parameters do satisfy the first, or the second, or both conditions of stability, respectively. Figures~\ref{fig:nonstab} and~\ref{fig:nonstab2} show examples of parameters satisfying the instability condition.

Figures~\ref{fig:Hopf1} and~\ref{fig:Hopf2} illustrate the evolution of the solution of the restricted model when crossing a Hopf bifurcation in the parameter space.

\begin{figure}[!htb]
    \centering
    \begin{tabular}{@{}cc@{}}
    \includegraphics[width=0.38\textwidth]{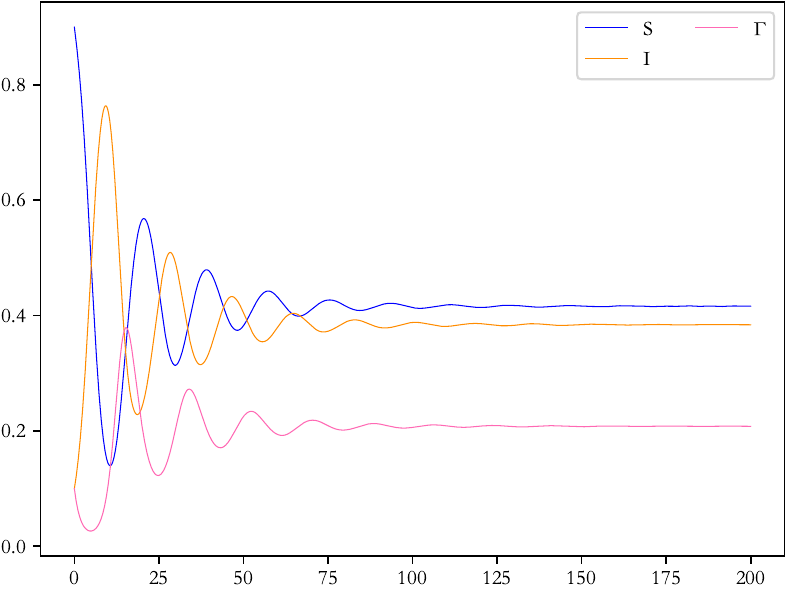}&\includegraphics[width=0.38\textwidth]{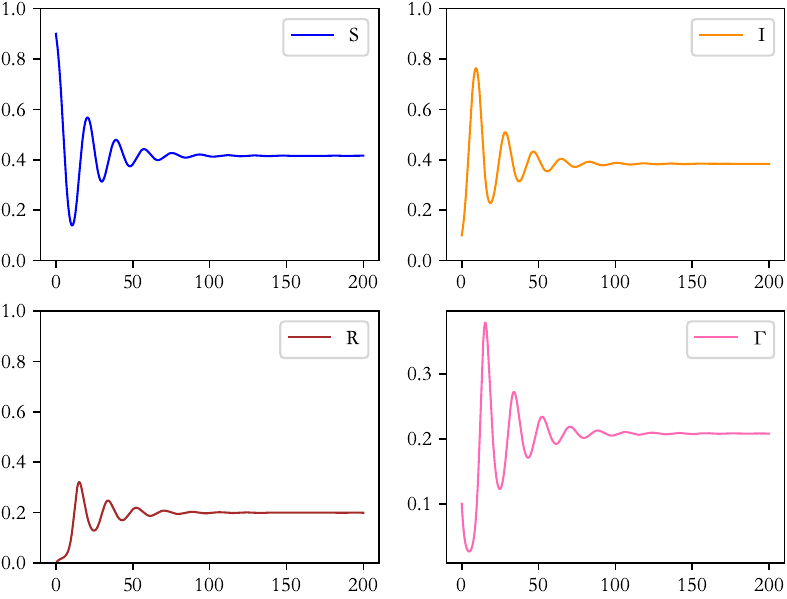}\\
    \includegraphics[width=0.38\textwidth]{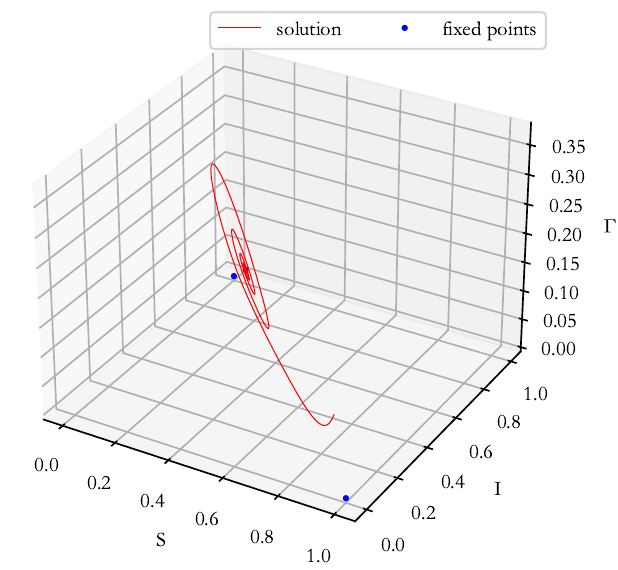}&\includegraphics[width=0.38\textwidth]{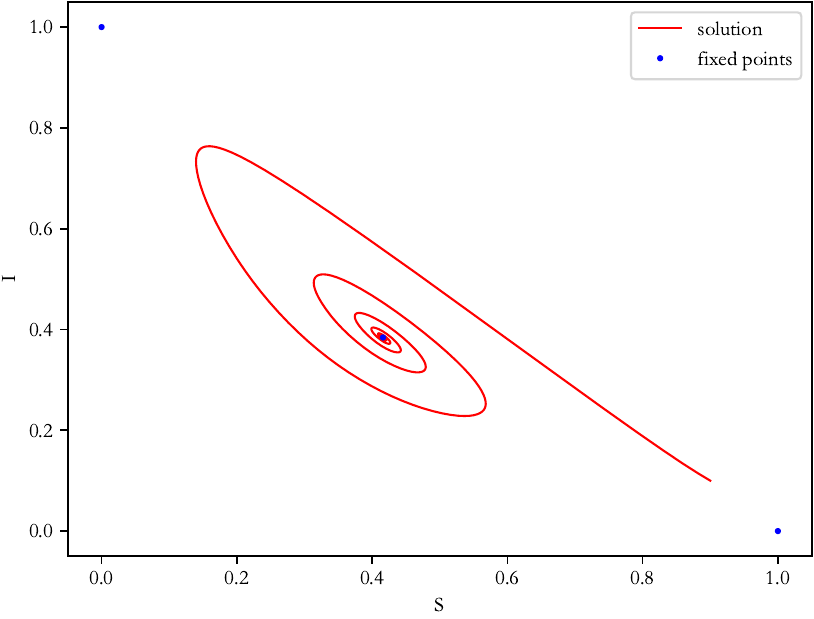}
    \end{tabular}
    \caption{SIRS model with feedback, with fixed point satisfying the first stability condition: $\alpha=0.65$, $\beta=0.5$, $\delta=0.6$, $\xi=0.4$, $S(0)=0.9$, $I(0)=\Gamma(0)=0.1$.}
    \label{fig:stab2}
\end{figure}
\begin{figure}[!htb]
    \centering
    \begin{tabular}{@{}cc@{}}
    \includegraphics[width=0.38\textwidth]{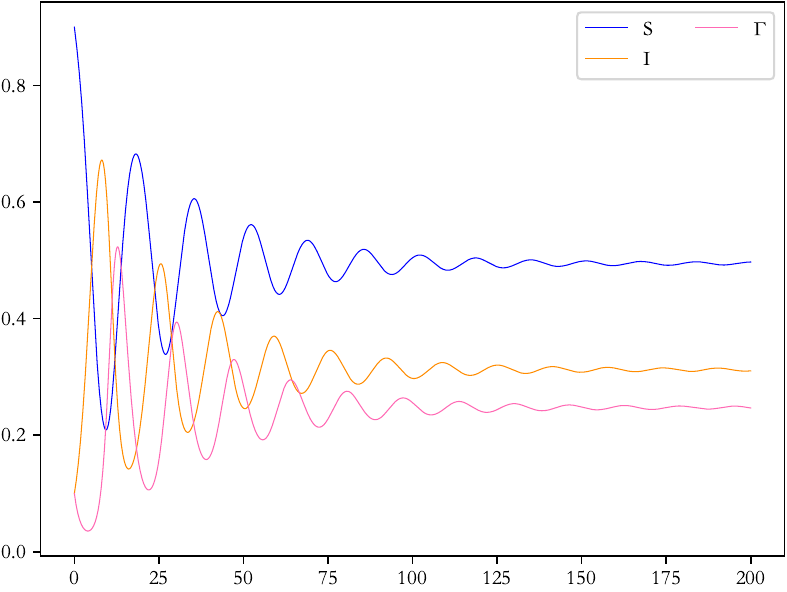}&\includegraphics[width=0.38\textwidth]{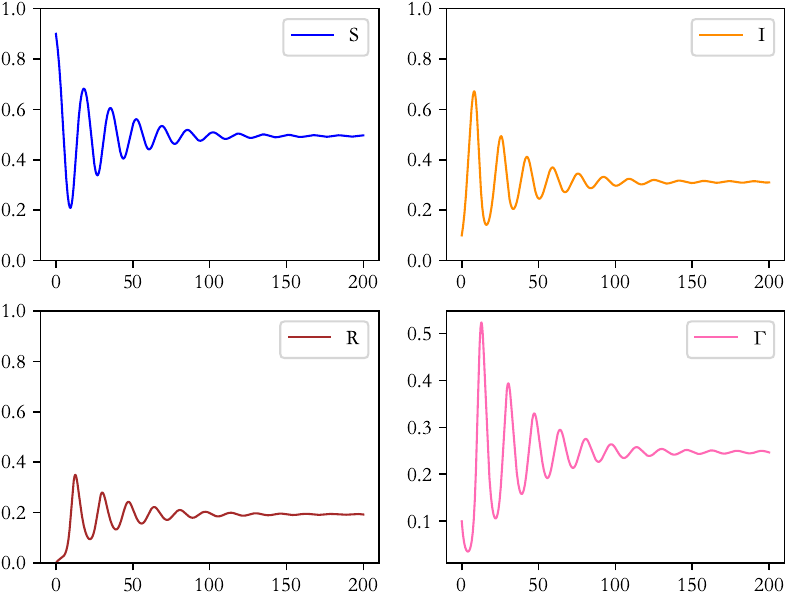}\\
    \includegraphics[width=0.38\textwidth]{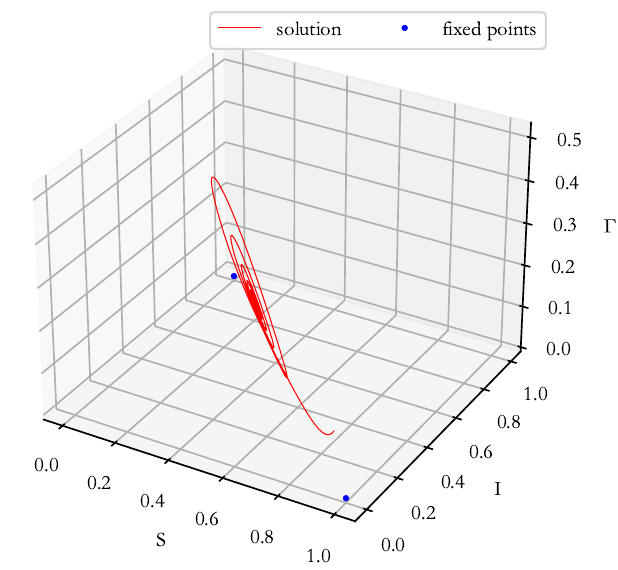}&\includegraphics[width=0.38\textwidth]{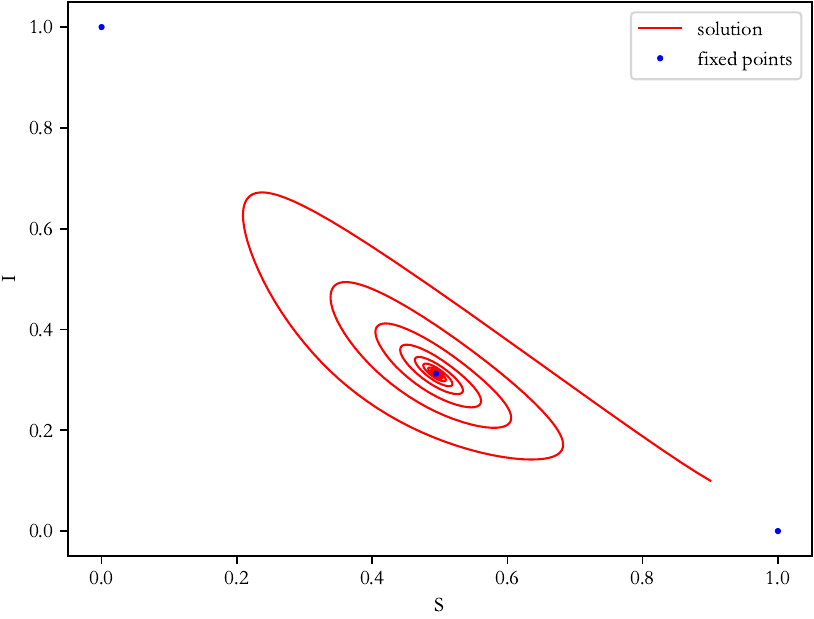}
    \end{tabular}
    \caption{SIRS model with feedback, with fixed point satisfying the second stability condition: $\alpha=0.95$, $\beta=0.5$, $\delta=0.6$, $\xi=0.4$, $S(0)=0.9$, $I(0)=\Gamma(0)=0.1$.}
    \label{fig:stab3}
\end{figure}
\begin{figure}[!htb]
    \centering
    \begin{tabular}{@{}cc@{}}
    \includegraphics[width=0.38\textwidth]{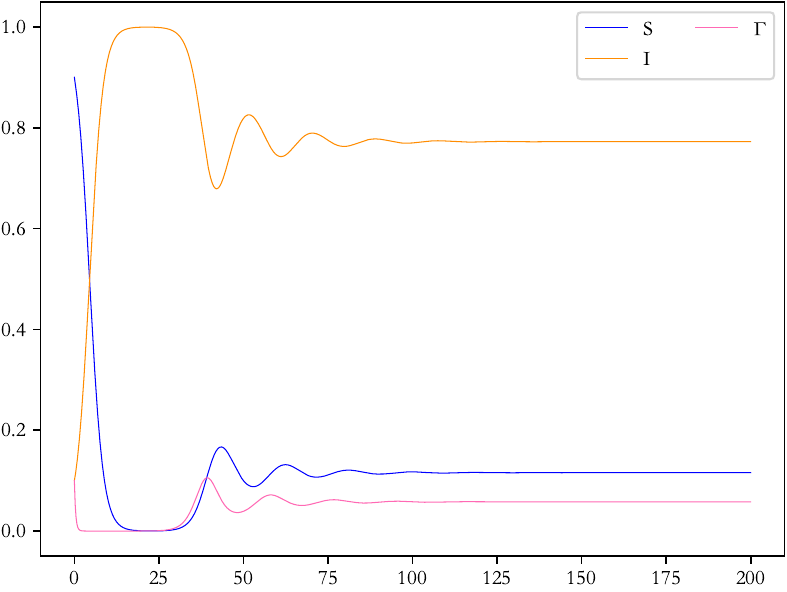}&\includegraphics[width=0.38\textwidth]{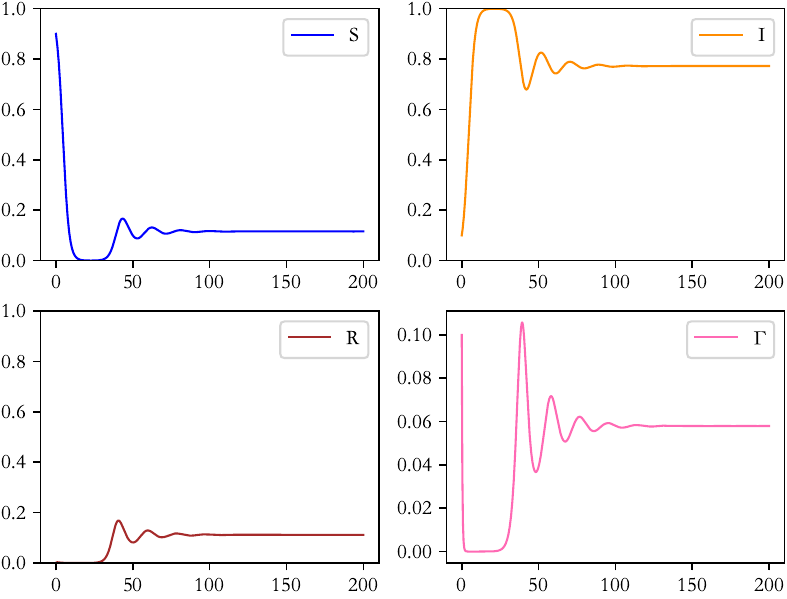}\\
    \includegraphics[width=0.38\textwidth]{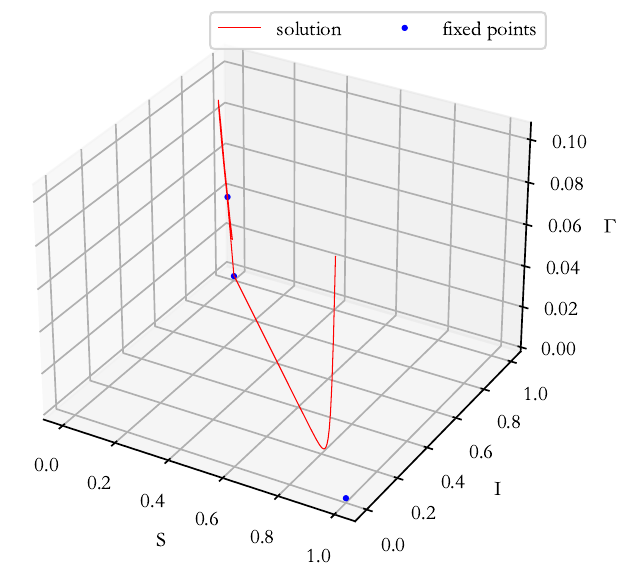}&\includegraphics[width=0.38\textwidth]{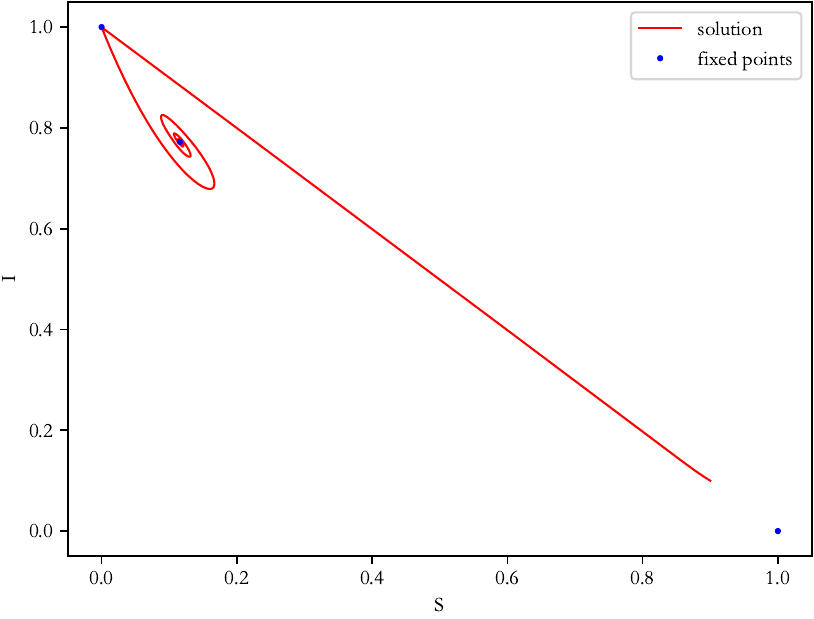}
    \end{tabular}
    \caption{SIRS model with feedback, with fixed point satisfying both stability conditions: $\alpha=0.45$, $\beta=0.5$, $\delta=3$, $\xi=0.4$, $S(0)=0.9$, $I(0)=\Gamma(0)=0.1$.}
    \label{fig:stab1}
\end{figure}
\begin{figure}[!htb]
    \centering
    \begin{tabular}{@{}cc@{}}
    \includegraphics[width=0.38\textwidth]{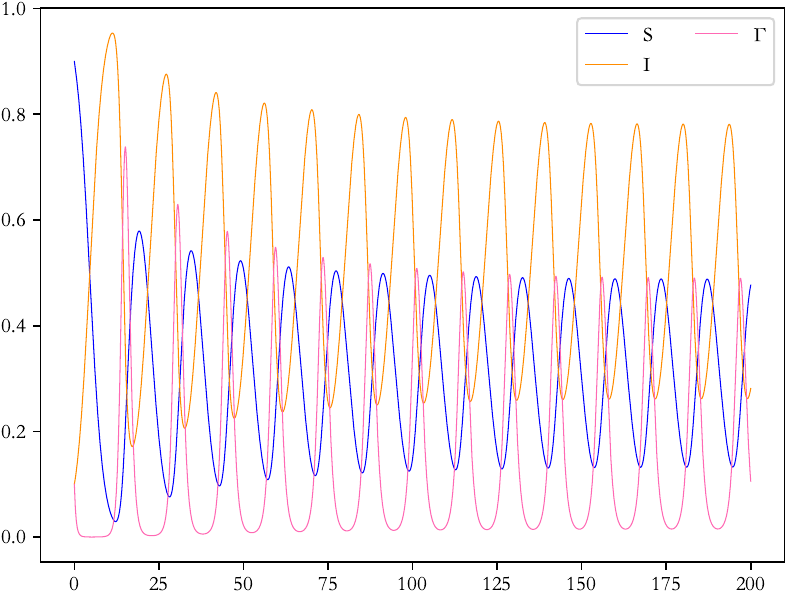}&\includegraphics[width=0.38\textwidth]{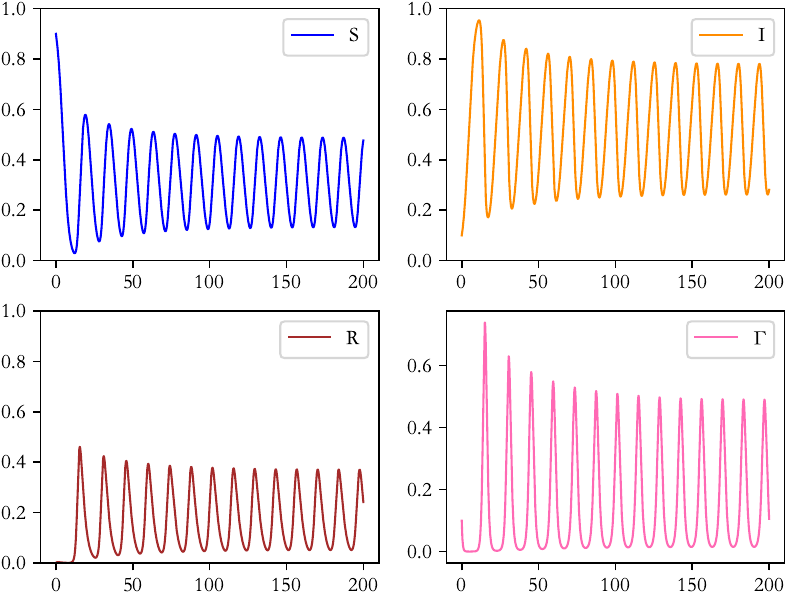}\\
    \includegraphics[width=0.38\textwidth]{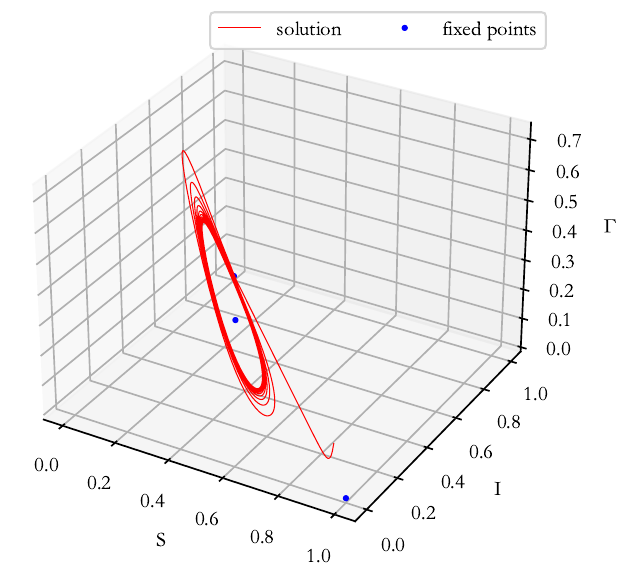}&\includegraphics[width=0.38\textwidth]{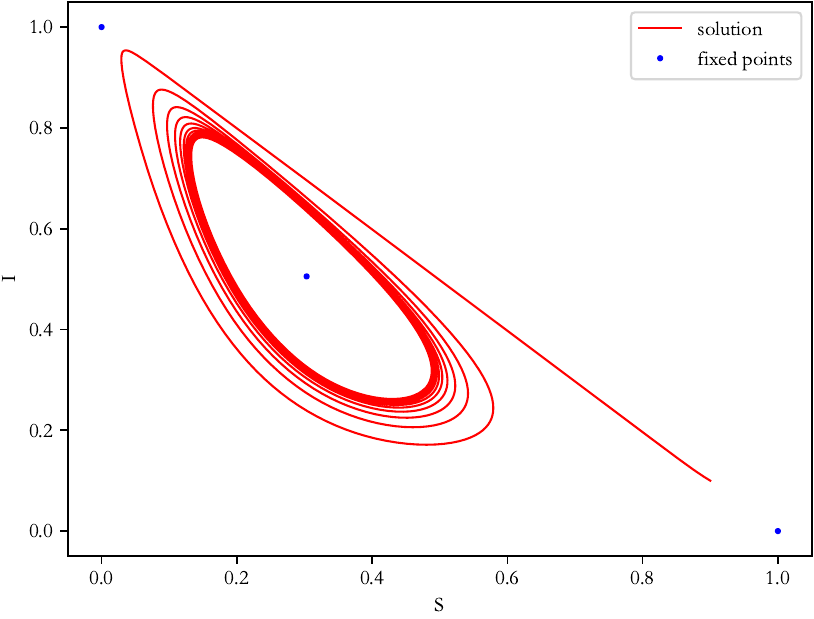}
    \end{tabular}
    \caption{SIRS model with feedback, with fixed point not satisfying the stability conditions: $\alpha=1.5$, $\beta=0.5$, $\delta=2.5$, $\xi=0.4$, $S(0)=0.9$, $I(0)=\Gamma(0)=0.1$.}
    \label{fig:nonstab}
\end{figure}
\begin{figure}[!htb]
    \centering
    \begin{tabular}{@{}cc@{}}
    \includegraphics[width=0.38\textwidth]{figs/SIRSdetHopf+0.25/sol.pdf}&\includegraphics[width=0.38\textwidth]{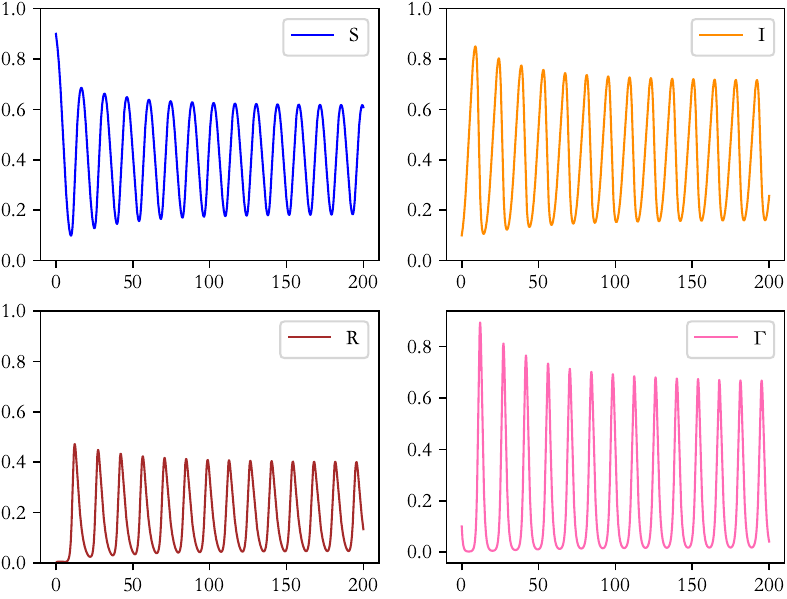}\\
    \includegraphics[width=0.38\textwidth]{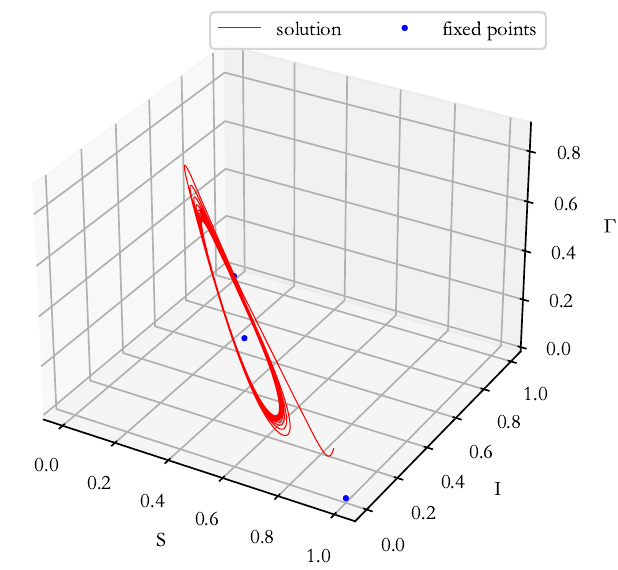}&\includegraphics[width=0.38\textwidth]{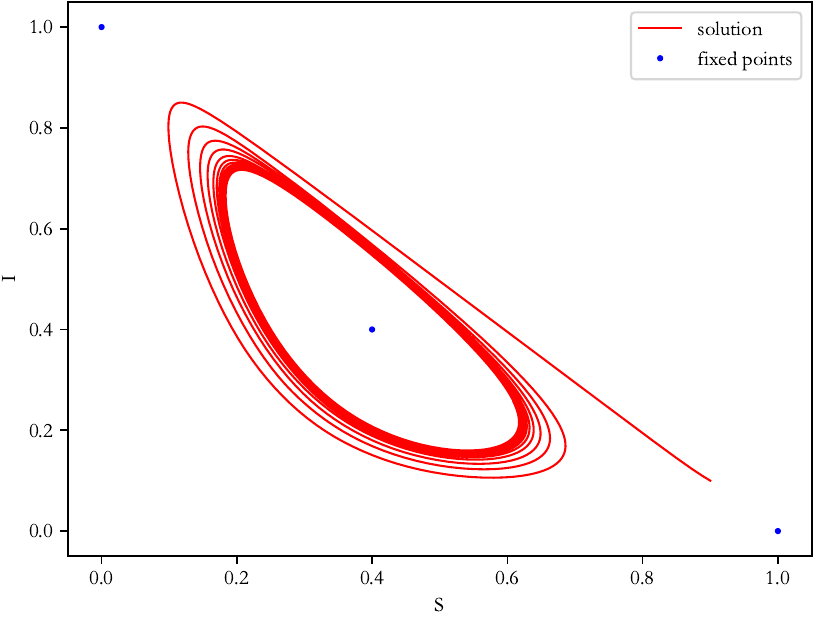}
    \end{tabular}
    \caption{SIRS model with feedback, with fixed point not satisfying the stability conditions: $\alpha=1.75$, $\beta=0.5$, $\delta=1.75$, $\xi=0.4$, $S(0)=0.9$, $I(0)=\Gamma(0)=0.1$.}
    \label{fig:nonstab2}
\end{figure}

\begin{figure}[!htb]
\centering
\begin{tabular}{@{}ccc@{}}
\includegraphics[width=0.3\textwidth]{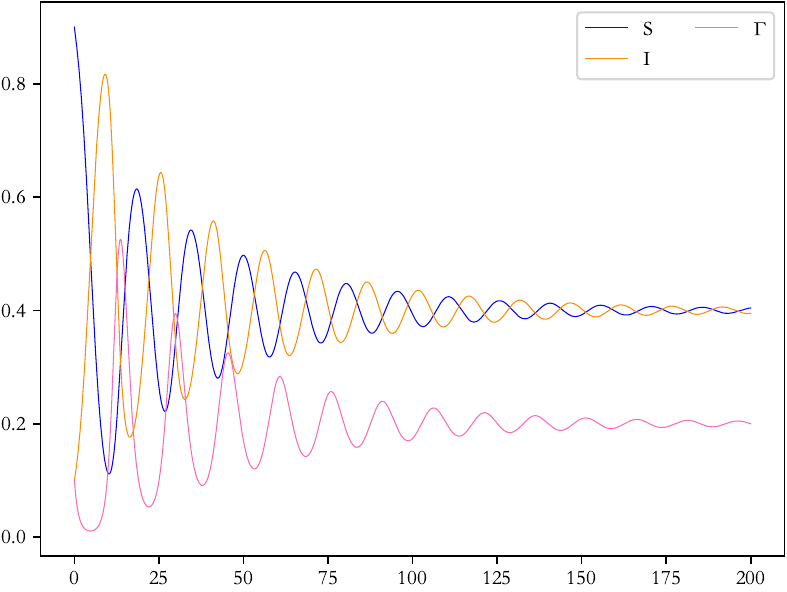}
&\includegraphics[width=0.3\textwidth]{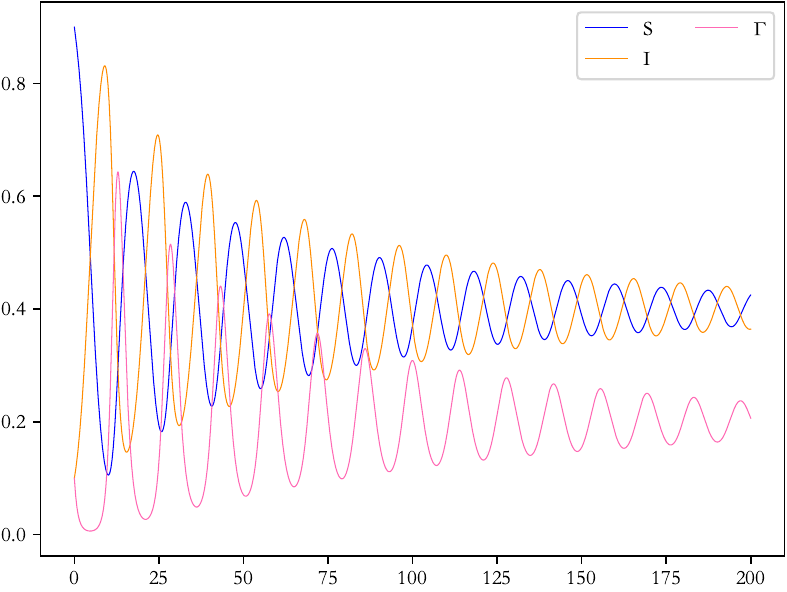}
&\includegraphics[width=0.3\textwidth]{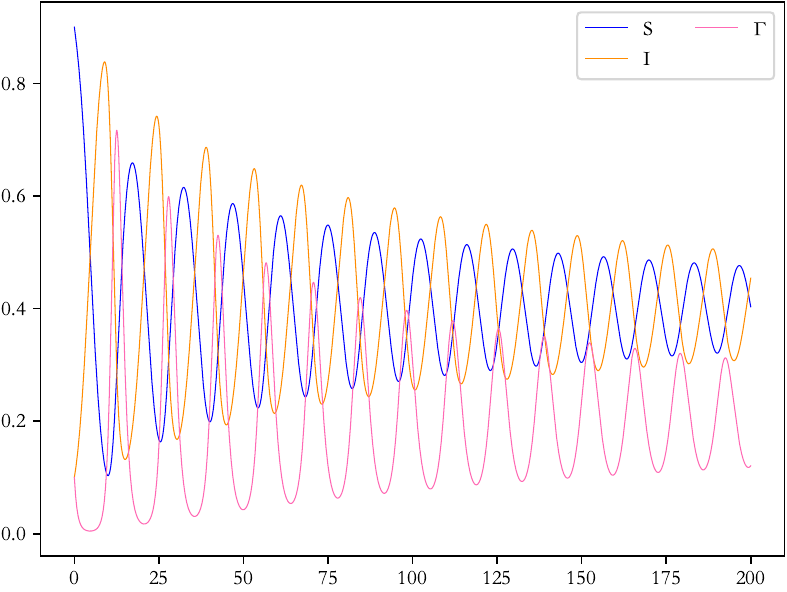}\\
$\alpha=\delta=1$ (stable)&$\alpha=\delta=1.25$ (stable)&$\alpha=\delta=1.4$ (stable)\\[3ex]
\includegraphics[width=0.3\textwidth]{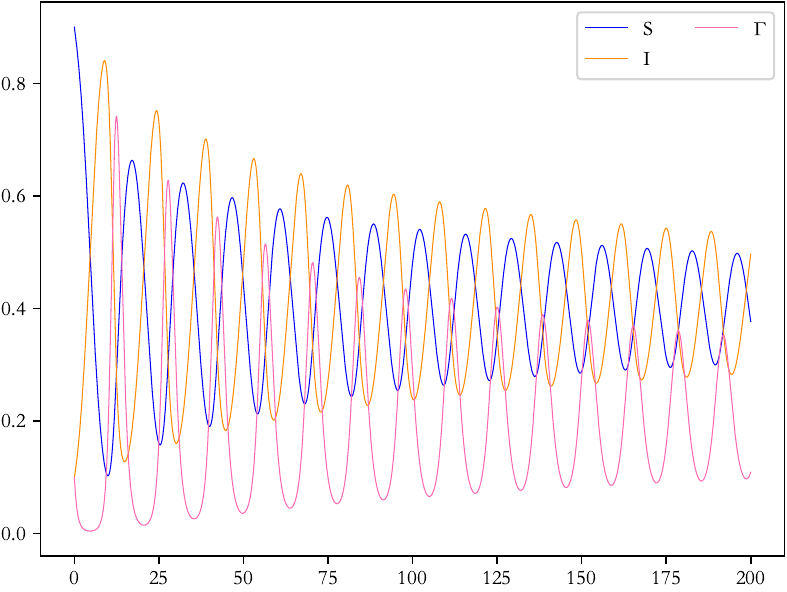}
&\includegraphics[width=0.3\textwidth]{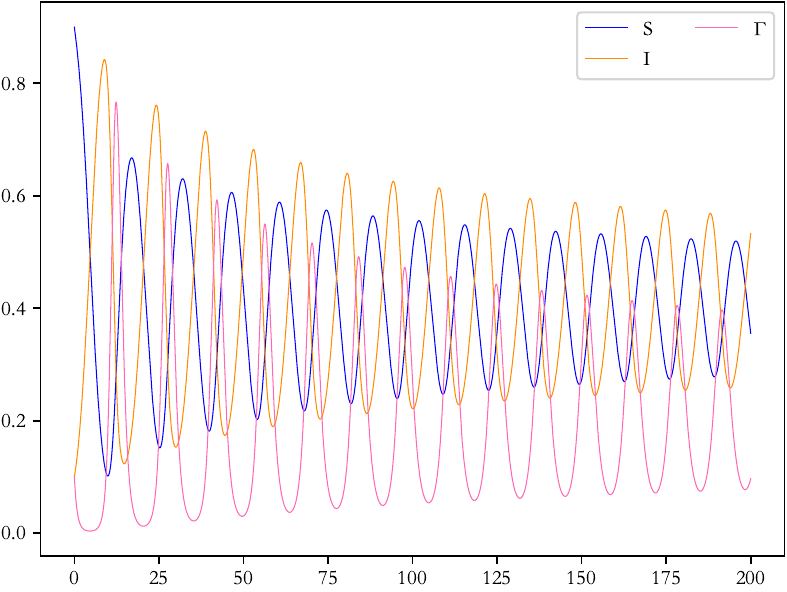}
&\includegraphics[width=0.3\textwidth]{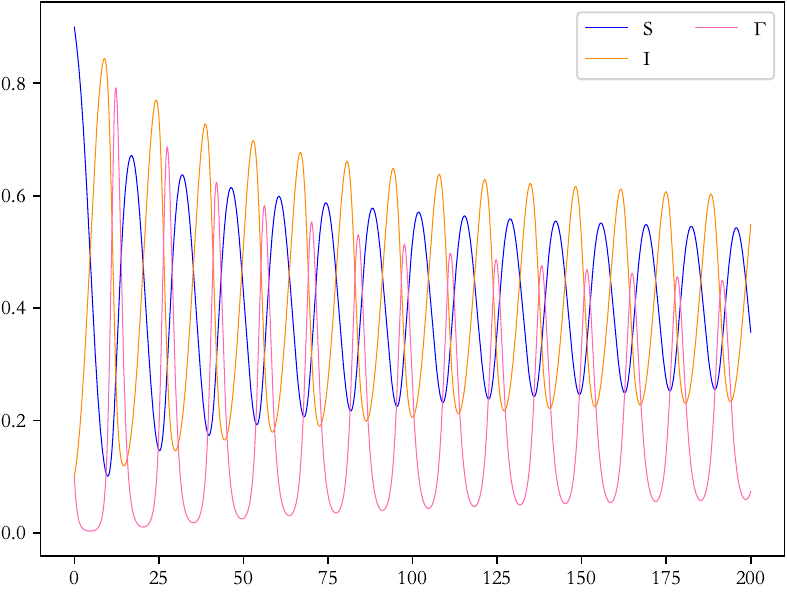}\\
$\alpha=\delta=1.45$ (stable)&$\alpha=\delta=1.5$ (Hopf bifurcation)&$\alpha=\delta=1.55$ (unstable)\\[3ex]
\includegraphics[width=0.3\textwidth]{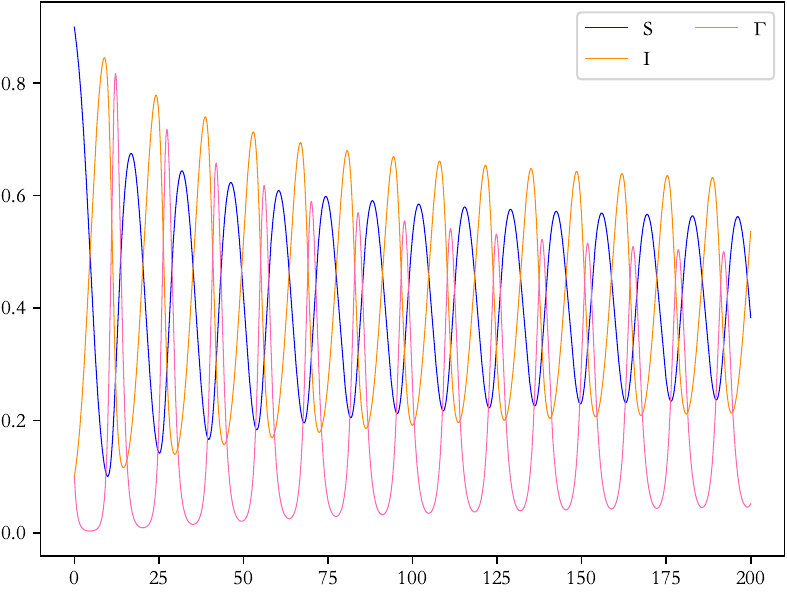}
&\includegraphics[width=0.3\textwidth]{figs/SIRSdetHopf+0.25/sol.pdf}
&\includegraphics[width=0.3\textwidth]{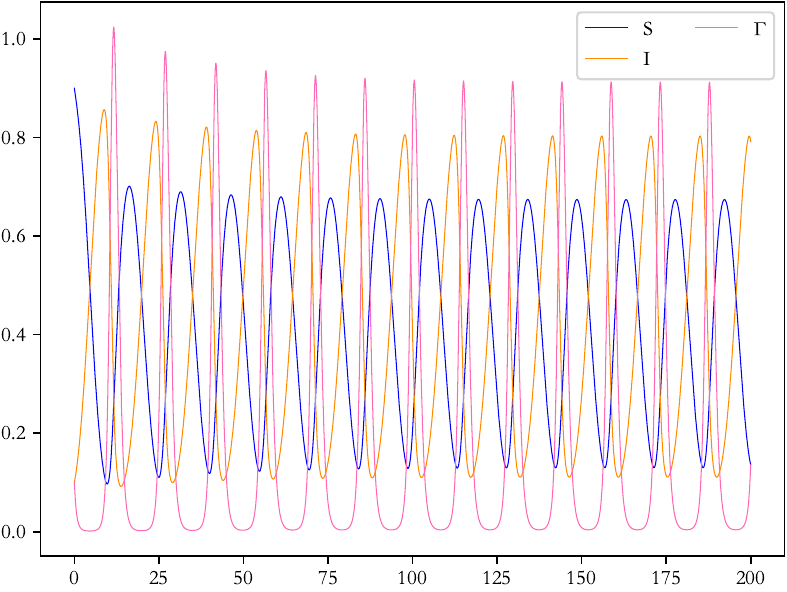}\\
$\alpha=\delta=1.6$ (unstable)&$\alpha=\delta=1.75$ (unstable)&$\alpha=\delta=2$ (unstable)
\end{tabular}
\caption{Evolution of the solution of the restricted version of the SIRS model with feedback when crossing a Hopf bifurcation point in the parameter space; $\beta=0.5$, $\xi=0.4$, $S(0)=0.9$, $I(0)=\Gamma(0)=0.1$, $\alpha=\delta$ from $1$ (upper left -- stable), to 1.5 (center -- Hopf bifurcation), to 2 (lower right -- unstable)}\label{fig:Hopf1}
\end{figure}

\begin{figure}[!htb]
\centering
\begin{tabular}{@{}ccc@{}}
\includegraphics[width=0.3\textwidth]{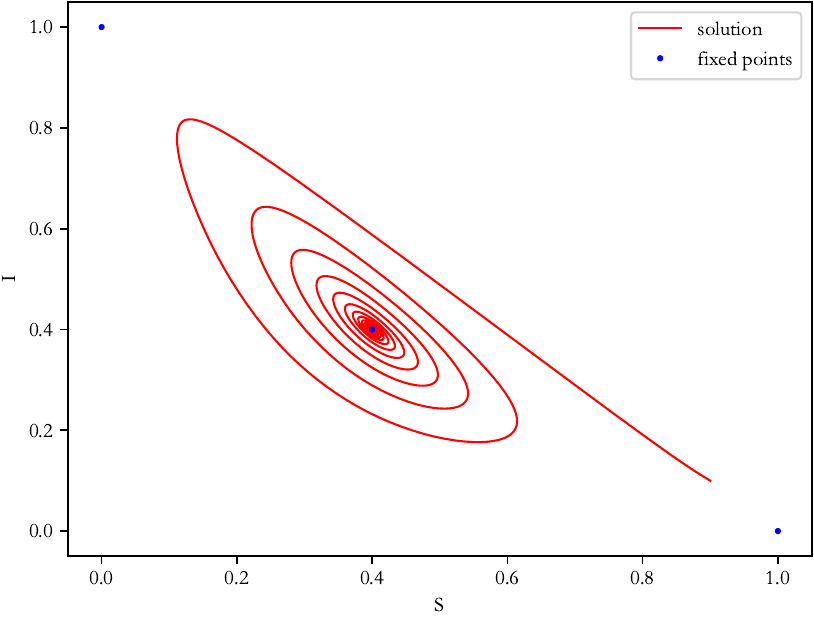}
&\includegraphics[width=0.3\textwidth]{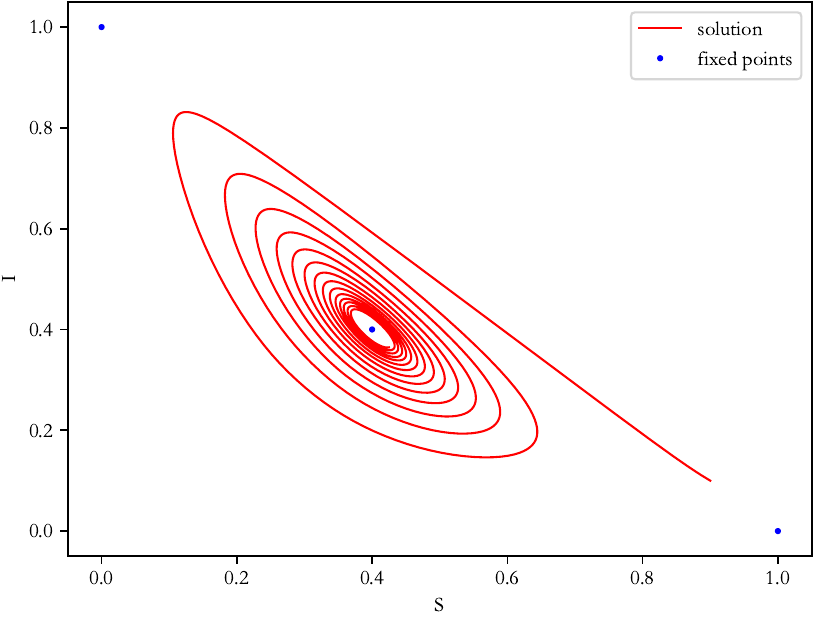}
&\includegraphics[width=0.3\textwidth]{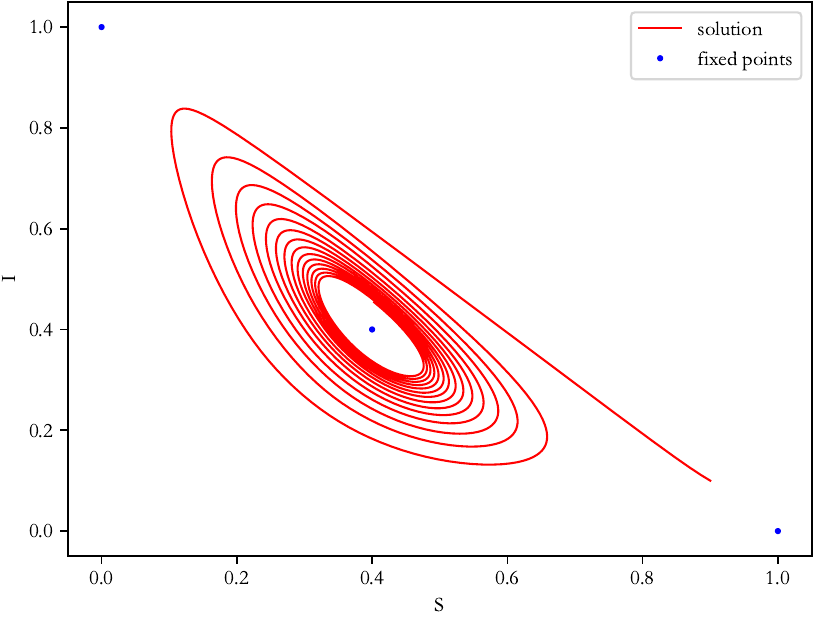}\\
$\alpha=\delta=1$ (stable)&$\alpha=\delta=1.25$ (stable)&$\alpha=\delta=1.4$ (stable)\\[3ex]
\includegraphics[width=0.3\textwidth]{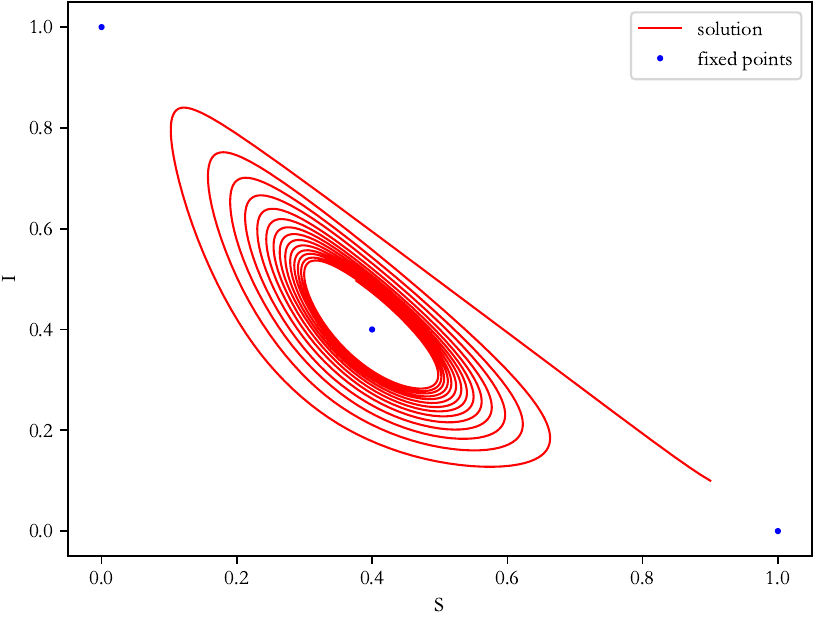}
&\includegraphics[width=0.3\textwidth]{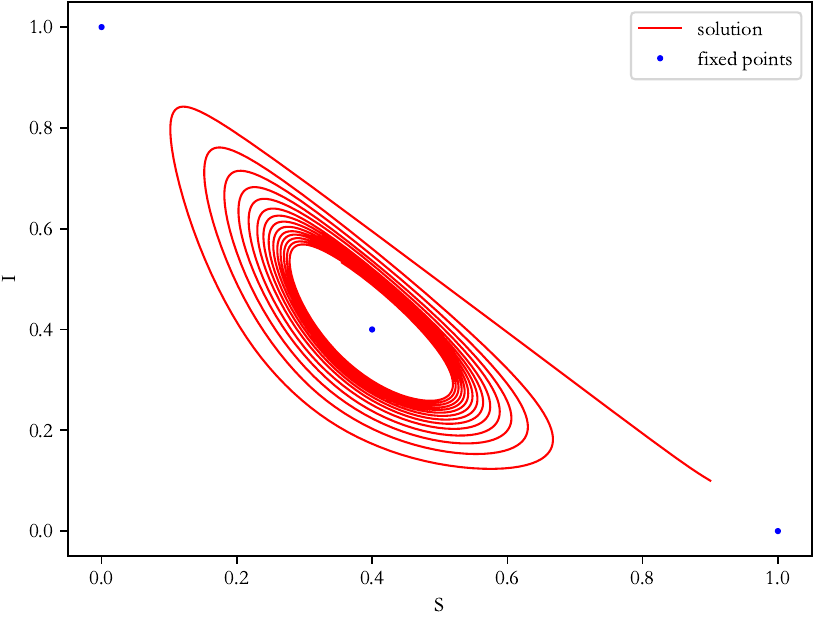}
&\includegraphics[width=0.3\textwidth]{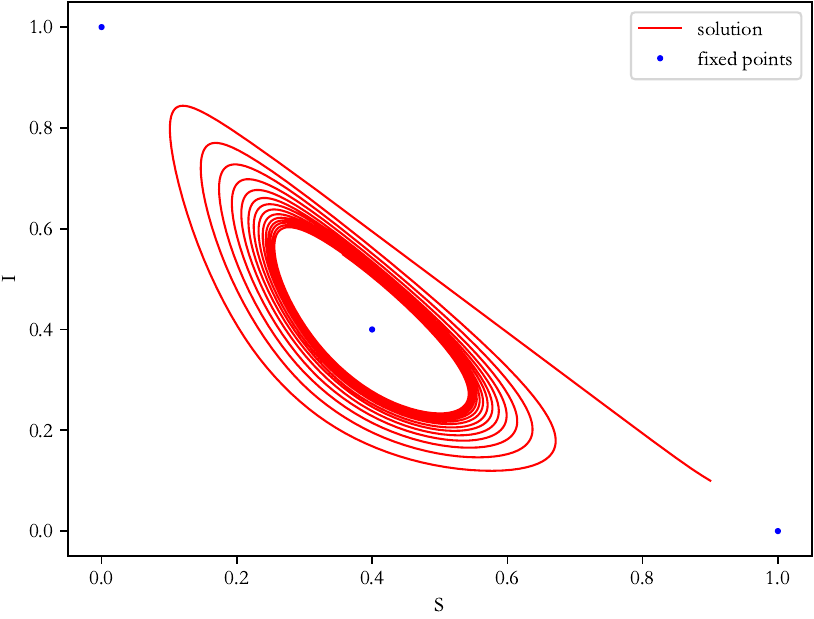}\\
$\alpha=\delta=1.45$ (stable)&$\alpha=\delta=1.5$ (Hopf bifurcation)&$\alpha=\delta=1.55$ (unstable)\\[3ex]
\includegraphics[width=0.3\textwidth]{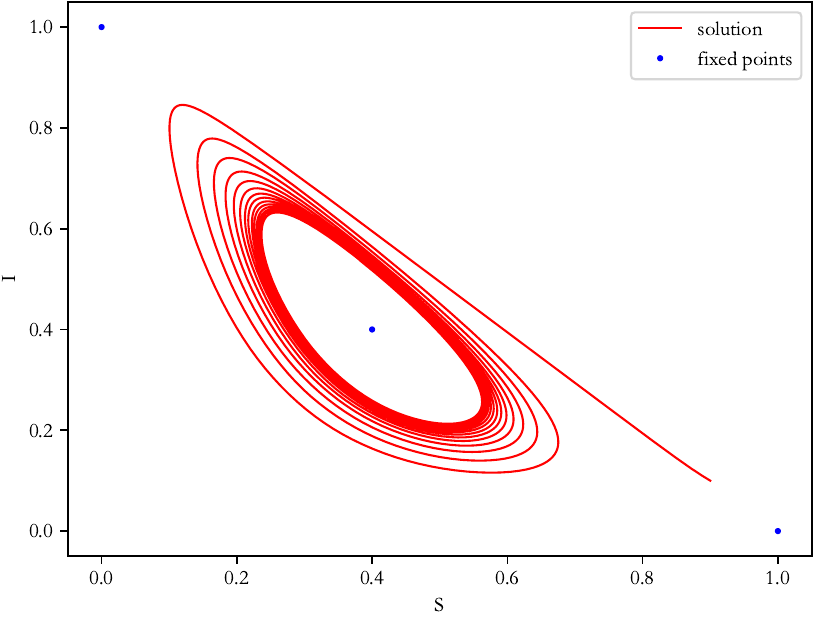}
&\includegraphics[width=0.3\textwidth]{figs/SIRSdetHopf+0.25/sol2d.pdf}
&\includegraphics[width=0.3\textwidth]{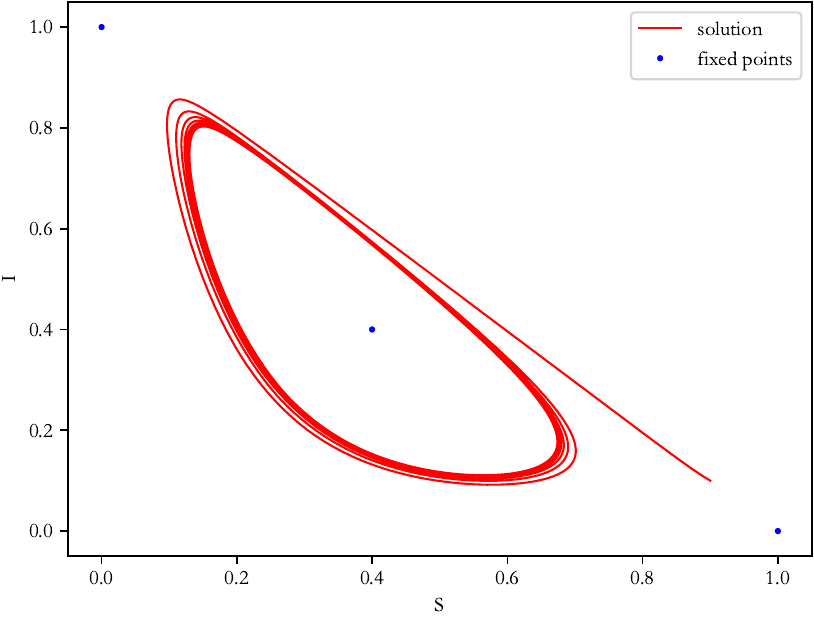}\\
$\alpha=\delta=1.6$ (unstable)&$\alpha=\delta=1.75$ (unstable)&$\alpha=\delta=2$ (unstable)
\end{tabular}
\caption{Evolution of the solution (projection on the $(S,I)$ plane) of the restricted version of the SIRS model with feedback when crossing a Hopf bifurcation point in the parameter space; $\beta=0.5$, $\xi=0.4$, $S(0)=0.9$, $I(0)=\Gamma(0)=0.1$, $\alpha=\delta$ from $1$ (upper left -- stable), to 1.5 (center -- Hopf bifurcation), to 2 (lower right -- unstable)}\label{fig:Hopf2}
\end{figure}

\clearpage

\section{Google Trends analysis}
\label{app:data}

The code used for the analysis is available on GitHub  \href{https://github.com/alessiomuscillo/modeling_waves}{at this link}.

\subsection{Seasonal decomposition}
\label{app:seasonal_decomposition}

The seasonal decomposition process uses the Python package \href{https://www.statsmodels.org/stable/generated/statsmodels.tsa.seasonal.seasonal_decompose.html}{\texttt{seasonal\_decompose}} and separates the time series additively into three components:
$$
y_t = T_t + S_t + R_t,
$$
where
\begin{itemize}
    \item   $y_t$ is the observed time series;
    \item   $T_t$ is the trend, that is, the underlying, slow-moving progression of the series over time;
    \item   $S_t$ is the seasonality, that is, the repeating short-term cycle or pattern within the data (e.g., annual or monthly seasonality);
    \item   $R_t$ is the residual, that is, the random noise or variation left over after removing the trend and seasonality.
\end{itemize}
The trend is computed as a local average over a rolling time window of length 52 weeks (52 weeks for annual seasonality). The seasonal component is computed by averaging the detrended series for each week of every year. For example, since we are analyzing weekly data with the assumption of a 52-week seasonality, the average seasonal effect is computed by grouping the same week of every year (e.g., the 4th week of every year) and averaging these de-trended data. 
The seasonal component is thus the pattern that remains after removing the trend.
The residual (or irregular) component is what is left after removing the trend and seasonality from the original time series, that is: $R_t = y_t - T_t - S_t$.
It represents random fluctuations, outliers, or noise that cannot be explained by the trend or seasonality. 
Lastly, we normalize the residuals in the range $[-1,1]$.

\subsection{Dynamic Time Warping (DTW)}

Dynamic Time Warping (DTW) is an algorithm used to measure the similarity between two time series, even if they are of different lengths or have temporal distortions (i.e., variations in speed or alignment). Unlike simple distance metrics such as Euclidean distance, which compares corresponding points in the series directly, DTW finds the optimal alignment between the two series by allowing for non-linear stretching or shrinking of the time axis. This flexibility makes DTW particularly useful for comparing time series that have similar overall patterns but are not perfectly synchronized. In our code, we used the \href{https://pypi.org/project/dtaidistance/}{\texttt{dtaidistance}} Python package which provides efficient implementations of DTW.

The DTW algorithm computes the distance by constructing a cost matrix, where each entry represents the cumulative distance between points from the two series. The optimal alignment path is then identified by minimizing the total accumulated distance along this matrix. The result is a DTW distance, where a smaller distance indicates higher similarity between the two series, and a larger distance suggests greater dissimilarity.

\subsection{Distance Computation for Residuals}
\label{app:distance_residuals}

For each word in our dataset, we computed the Dynamic Time Warping (DTW) distance between the residuals of the word and two reference processes: our endogenous oscillation model and a set of random walks.

\begin{itemize}
    \item   \textbf{Distance to our Model.} 
    To compare the residuals to our model, we first generated a family of models where the parameter $\beta$, which controls the characteristics of the oscillation, was varied over a specified range $[0.01, 0.3]$. The other parameters of the model are set as: $\xi = 0.1$ and $\alpha = \beta+\xi+\sqrt{\xi(\beta+\xi)} + 0.01$. 
    For each word, we computed the DTW distance between its residuals and the model corresponding to each $\beta$. We then selected the model with the $\beta$ that minimized this distance, ensuring the closest match between the residuals and our model of endogenous oscillations. This process allowed us to quantify how well the model captures the behavior of each word's residuals after removing the trend and seasonality components.
    \item   \textbf{Distance to Random Walks.} 
    For comparison, we simulated 500 random walks (based on normal distributions with mean equal to 0 and variance equal to 1), each of the same length as the word's residuals. For each random walk, we computed the DTW distance to the residuals and then averaged these distances to obtain a single value representing the similarity of the word's residuals to a typical random walk. This comparison provides a baseline for determining whether the residuals follow a structured, oscillatory behavior or are better explained by random fluctuations.
\end{itemize}
By comparing the DTW distances to both the model and the random walks, we can assess whether the residuals exhibit patterns that are more consistent with endogenous oscillations than with random fluctuations.

\end{document}